%% file: main.tex
\documentclass[runningheads]{llncs}
\ifdefined\fullversion
\else
\def\fullversion{0}    
\fi

 \def\anonymous{0}      

\input{tex_files/ZZ_header}
\input{tex_files/ZZfont_header}

\title{
Regular Syndrome Decoding Is Easier Than Previously Estimated 
}
\titlerunning{
RSD Is Easier Than Previously Estimated
}

\ifnum\anonymous=0
\author{
  Jiseung Kim\inst{1} \and
  Minkyu Kim\inst{2} \and
  Changmin Lee\inst{3} \and
  Hochang Lee\inst{2}
}
\institute{Jeonbuk National University \\
  \href{mailto:jiseungkim@jbnu.ac.kr}{jiseungkim@jbnu.ac.kr} \and
  The Affiliated Institute of ETRI \\
  \href{mailto:mkkim@nsr.re.kr}{mkkim@nsr.re.kr}, 
  \href{mailto:mkkim@nsr.re.kr}{lhc254@nsr.re.kr} \and
  Korea University \\
  \href{mailto:changminlee@korea.ac.kr}{changminlee@korea.ac.kr}
}
\else
\author{} 
\institute{}
\fi

\begin{document}
  \maketitle

\begin{abstract}
The Regular Syndrome Decoding (RSD) problem is a structured variant of classical Syndrome Decoding and underlies a growing number of efficient cryptographic constructions, including zero-knowledge protocols, pseudorandom correlation generators, and MPC-in-the-Head signatures. Although RSD is generally believed to be computationally hard, recent works have shown that algebraic attacks can outperform information-set decoding on important RSD parameter regimes.

Building on this line of work, we revisit the concrete hardness of RSD and show that many widely used parameter sets are substantially easier to solve than previously estimated. 
This improvement stems from a refined hybrid-algebraic attack that exposes and removes
a previously overlooked source of concrete complexity: the joint blockwise placement of zero-coordinate guesses and syndrome-based variable
eliminations. Rather than treating the two guessing stages and linear
elimination as largely separate choices, our method coordinates them
across all regular blocks so that the guessing penalty is balanced while
sufficient equations are retained for linearization. This gives a systematic
block-balancing principle for the entire guess–reduce–solve pipeline. 

We reassess several RSD-based applications, including Ferret (CCS~2020), Wolverine (IEEE S\&P~2021), and committed VOLE (CCS~2025). Our estimates place at least one parameter set per application below its 128-bit target, with losses approaching 67~bits. On PCG benchmarks, we improve previous best attack costs by up to 29.77 and 25.91~bits over non-binary and binary fields, respectively. Moreover, for the revised 128-bit PCG parameter sets of Liu et al. (EUROCRYPT~2024) evaluated here, estimated attack costs are all below $2^{111}$ field operations.
These results substantially revise the concrete-security landscape of RSD and indicate that several existing parameter choices overestimate the hardness of their underlying RSD instances.

\end{abstract}

\input{tex_files/body}

\ifnum\fullversion=0
    \bibliographystyle{bib_files/splncs04}
\else
    \bibliographystyle{bib_files/alpha-short}
\fi
\bibliography{bib_files/abbrev3,bib_files/add}

\appendix
\renewcommand{\thesection}{Supplementary \Alph{section}} 
\renewcommand{\thefigure}{\thesection.\arabic{figure}}   
\setcounter{figure}{0}
\newpage
\input{tex_files/Appendix}
\end{document}

%% file: tex_files/ZZ_header.tex
\usepackage[utf8]{inputenc}
\usepackage[T1]{fontenc}
\usepackage{hyperref}
\usepackage{verbatim}
\usepackage{tikz}
\usetikzlibrary{positioning, arrows.meta, calc, fit, backgrounds, decorations.pathreplacing}

\usepackage{makecell}

\definecolor{coralfill}{HTML}{FAECE7}
\definecolor{coralstroke}{HTML}{993C1D}
\definecolor{tealfill}{HTML}{E1F5EE}
\definecolor{tealstroke}{HTML}{0F6E56}
\definecolor{grayfill}{HTML}{F1EFE8}
\definecolor{graystroke}{HTML}{5F5E5A}
\definecolor{purplefill}{HTML}{EEEDFE}
\definecolor{purplestroke}{HTML}{534AB7}

\tikzset{
  block/.style={rounded corners=4pt, minimum height=1cm, draw, thick, font=\small},
  coral/.style={block, fill=coralfill, draw=coralstroke},
  teal/.style={block, fill=tealfill, draw=tealstroke},
  gray/.style={block, fill=grayfill, draw=graystroke},
  purple/.style={block, fill=purplefill, draw=purplestroke},
  arr/.style={-{Stealth[length=5pt]}, thick, gray!70!black},
  dashedbox/.style={rounded corners=4pt, draw=gray!50, dashed, thick},
}

\usepackage{xspace}
\usepackage{amsmath} 
\usepackage{amssymb}
\usepackage{mathtools}
\usepackage{pifont}
\usepackage{etoolbox}
\usepackage[normalem]{ulem}
\usepackage{booktabs}
\usepackage{array}
\usepackage[capitalise,noabbrev]{cleveref}
\usepackage{cite}
\usepackage{multibib}
\usepackage{url}
\usepackage{algorithm}
\usepackage{algpseudocode}
\usepackage{paralist}
\usepackage{mathrsfs}
\usepackage{relsize}
\usepackage{stmaryrd}
\usepackage[textsize=small]{todonotes}
\usepackage{multirow}
\usepackage[lambda,n,operators]{cryptocode}

\usepackage{graphicx}
\usepackage{kotex}

\newtoggle{notes}
\toggletrue{notes} 

\newcommand{\fix}{\textup{pre}}

\newcommand{\Z}{\mathbb{Z}}
\newcommand{\F}{\mathbb{F}}
\newcommand{\N}{\mathbb{N}}

\makeatletter
\DeclareFontFamily{OMX}{MnSymbolE}{}
\DeclareSymbolFont{MnLargeSymbols}{OMX}{MnSymbolE}{m}{n}
\SetSymbolFont{MnLargeSymbols}{bold}{OMX}{MnSymbolE}{b}{n}
\DeclareFontShape{OMX}{MnSymbolE}{m}{n}{
    <-6>  MnSymbolE5
   <6-7>  MnSymbolE6
   <7-8>  MnSymbolE7
   <8-9>  MnSymbolE8
   <9-10> MnSymbolE9
  <10-12> MnSymbolE10
  <12->   MnSymbolE12
}{}
\DeclareFontShape{OMX}{MnSymbolE}{b}{n}{
    <-6>  MnSymbolE-Bold5
   <6-7>  MnSymbolE-Bold6
   <7-8>  MnSymbolE-Bold7
   <8-9>  MnSymbolE-Bold8
   <9-10> MnSymbolE-Bold9
  <10-12> MnSymbolE-Bold10
  <12->   MnSymbolE-Bold12
}{}

\let\llangle\@undefined
\let\rrangle\@undefined
\DeclareMathDelimiter{\llangle}{\mathopen}%
                     {MnLargeSymbols}{'164}{MnLargeSymbols}{'164}
\DeclareMathDelimiter{\rrangle}{\mathclose}%
                     {MnLargeSymbols}{'171}{MnLargeSymbols}{'171}
\makeatother


%% file: tex_files/ZZfont_header.tex
\newcommand{\txP}{\textup{P}}
\newcommand{\txT}{\textup{T}}

\newcommand{\calF}{\mathcal F}

\newcommand{\calO}{\mathcal O}

\newcommand{\bfA}{\mathbf A}

\newcommand{\bfD}{\mathbf D}
\newcommand{\bfE}{\mathbf E}

\newcommand{\bfH}{\mathbf H}
\newcommand{\bfI}{\mathbf I}

\newcommand{\bfL}{\mathbf L}
\newcommand{\bfM}{\mathbf M}

\newcommand{\bfO}{\mathbf O}

\newcommand{\bfQ}{\mathbf Q}

\newcommand{\bfa}{\mathbf a}
\newcommand{\bfb}{\mathbf b}
\newcommand{\bfc}{\mathbf c}
\newcommand{\bfd}{\mathbf d}
\newcommand{\bfe}{\mathbf e}

\newcommand{\bfs}{\mathbf s}
\newcommand{\bft}{\mathbf t}

\newcommand{\bfx}{\mathbf x}

\newcommand{\bfz}{\mathbf z}

\newcommand{\sfB}{\mathsf B}

\newcommand{\sft}{\mathsf t}

%% file: tex_files/body.tex

\section{Introduction}\label{sec:introduction}

The regular syndrome decoding (RSD) problem is a structured variant of the classical syndrome decoding (SD) problem and has emerged as an increasingly important hardness assumption in modern cryptography. Its appeal comes from the additional structure imposed on the error vector: rather than allowing the nonzero coordinates to occur at arbitrary positions, RSD requires them to be distributed regularly across prescribed blocks. This structure can be exploited by protocol designers to obtain substantially more efficient constructions than those based on an unstructured decoding problem, while retaining a decoding-based security assumption.

As a result, RSD and closely related regular-noise assumptions now appear in a broad range of cryptographic primitives and protocols. They underlie pseudorandom correlation generators (PCGs), including constructions for efficiently generating correlated randomness and multiparty Beaver triples~\cite{CCS:BCGI18,C:KPRR25,AC:MMST25}; correlated oblivious transfer and VOLE-based protocols, including Ferret and committed VOLE~\cite{CCS:YWLZW20,CCS:SLYYWW25}; zero-knowledge protocols such as Wolverine~\cite{SP:WYKW21}; and secure multiparty computation frameworks such as TinyKeys~\cite{C:HOSS18,JC:HOSS22}. RSD has also become particularly relevant to post-quantum signatures, where it is used in several constructions based on the MPC-in-the-Head paradigm~\cite{EC:CarCouJou23,NISTPQC-ADD-R2:SDitH24,PKC:CLYYYZ24,AC:BCCGJ24,AC:OuyTanXu24}.Thus, the concrete hardness of RSD is no longer relevant to a single cryptographic construction: security estimates for RSD directly affect the parameter selection and security claims of a diverse collection of
cryptographic systems.

This growing use of RSD makes its concrete security particularly important. A moderate improvement in the best known algorithm for solving RSD can translate directly into a loss of security for several seemingly unrelated protocols that instantiate the same underlying hardness assumption. Consequently, understanding the precise cost of solving concrete RSD instances, rather than considering only its asymptotic hardness, is essential for assessing the security margins of existing constructions.

In the standard syndrome decoding problem, one is given a linear code defined by a parity-check matrix $\bfH \in \F_q^{(n-k) \times n}$ and a syndrome vector $\bfs \in \F_q^{n-k}$. The goal is to find an error vector $\bfx \in \F_q^n$ such that $\bfH \bfx = \bfs$ and $\bfx$ has Hamming weight $h$, i.e., exactly $h$ of its coordinates are nonzero.
The RSD problem imposes an additional regularity constraint on the support of $\bfx$. The coordinate set ${1,\ldots,n}$ is partitioned into $h$ disjoint blocks of size $\beta := n/h \in \Z$, and exactly one coordinate in each block is required to be nonzero. Thus, whereas standard syndrome decoding allows the support to be chosen arbitrarily among all weight-$h$ vectors, RSD restricts the support to one position per block. This additional combinatorial structure is the defining feature that distinguishes RSD from ordinary syndrome decoding and is the main source of the specialized algebraic techniques considered in this work.


\newpage
\subsection{Algebraic attacks on the RSD problem}

Information-set decoding (ISD) and its variants are the standard combinatorial tools to solve syndrome decoding problems~\cite{IRE:Pra62,CTA:Ste89,AC:MayMeuTho11,EC:BJMM12}. For RSD, however, algebraic attacks can become competitive and can even outperform ISD in some low code rate parameter regimes because the regularity constraints admit a compact polynomial representation. In particular, if two coordinates belong to the same block, regularity implies quadratic relations of the form $x_i x_j=0$~\cite{EC:BriOyg23}. Together with the linear syndrome equation, these relations define a multivariate polynomial system whose solution reveals the regular error vector.

The direct algebraic attack quickly becomes expensive as the number of variables grows. Consequently, the most effective attacks use a hybrid strategy: they first guess part of the support information, reduce the number of variables with the syndrome equations, and then solve the resulting polynomial system by algebraic techniques such as linearization or Gr\"obner-basis methods. The central concrete trade-off is therefore not only how many coordinates are guessed, but also where those guesses and eliminations occur across the regular blocks.

Wang et al.~\cite{AC:WWYLYZ25} recently gave the strongest hybrid algebraic attack for many RSD parameters considered in this work. Their method follows a guess--reduce--solve paradigm: it first guesses error-free positions, uses Gaussian elimination on the syndrome equation to eliminate variables, and then applies a second hybrid algebraic step to the reduced quadratic system. The second guessing stage, however, is concentrated in a restricted set of blocks, while the variables eliminated by the syndrome equations are also allocated unevenly. As we show below, this blockwise allocation creates an avoidable loss in guessing probability and leaves substantial room to reduce the concrete attack complexity.

\subsection{Our contribution}

Our main contribution is a stronger hybrid-algebraic algorithm for concrete RSD instances. The key novelty is a new way to \emph{coordinate the existing primitives}: zero guessing, syndrome-based elimination, preprocessing, and the final algebraic solve are treated as coupled blockwise operations rather than optimized in isolation. This viewpoint reveals that a substantial part of the complexity of prior attacks comes from an avoidable imbalance in where these operations are placed across the regular blocks.

\subsubsection{The source of the previous complexity loss}
A hybrid algebraic attack benefits from guessing error-free coordinates because the resulting polynomial system contains fewer variables. In RSD, however, each zero guess also removes quadratic equations of the form $x_i x_j=0$ from the same block.  Consequently, the cost of the attack is governed by a trade-off between the probability of a correct guess and the number of equations that remain available after the guess.

This trade-off depends strongly on how the guesses are distributed among the regular blocks.  For example, guessing two zero coordinates in one block succeeds with probability $1-2/\beta$ and removes $2\beta-3$ quadratic equations, whereas guessing one zero coordinate in each of two distinct blocks succeeds with probability $(1-1/\beta)^2>1-2/\beta$ while removing $2(\beta-1)$ equations. Thus, concentrating guesses may preserve slightly more equations, but it can incur a substantially larger guessing penalty.

Re-examining the attack of Wang et al.~\cite{AC:WWYLYZ25},
we identify exactly this effect at scale: the second zero-guessing stage and the syndrome-based eliminations occupy an unbalanced subset of blocks. Because these operations compete for the same blockwise coordinates, their interaction amplifies the guessing penalty. This is a structural inefficiency of the attack schedule, rather than a limitation of linearization itself.

\subsubsection{A globally balanced guess--reduce--solve technique}
Our main algorithmic contribution is a globally balanced organization of the entire guess--reduce--solve pipeline. Previous approaches optimize or place the individual stages largely in isolation, which can lead to a substantial loss in the overall guessing probability. In contrast, our algorithm jointly coordinates all operations that determine the final polynomial system: the first zero guesses, the variables eliminated using the syndrome equations, the variables involved in preprocessing, and the second zero guesses. These operations are distributed across the regular blocks as evenly as the discrete parameters permit.

The key point is that the gain does not come from simply increasing or decreasing the number of guessed coordinates, but from changing their interaction with syndrome elimination at the block level. By balancing the syndrome eliminations, our method preserves more uniformly sized candidate sets for the subsequent zero-guessing stage. The AM--GM principle underlying our probability analysis then shows that this joint balancing improves the overall success probability while retaining sufficiently many equations for linearization. Thus, the proposed strategy removes a source of concrete complexity that remains even when the individual ingredients of the hybrid attack are already well optimized.

We complement this algorithmic improvement with two refinements that further sharpen the concrete complexity. First, we revisit the preprocessing cost and identify operations within the preprocessing procedure that are essentially free. Explicitly accounting for these operations yields a tighter estimate of the actual preprocessing cost, without changing the underlying preprocessing algorithm. Second, we exploit the overdetermined structure of the final random linear system more precisely, allowing incorrect guesses to be rejected with fewer excess equations.

These contributions improve the attack at complementary levels. The global blockwise balancing substantially increases the probability of reaching the correct final system, while the refined complexity analysis reduces the estimated algebraic cost of processing each resulting system. Their combination leads to the significant concrete complexity reductions reported in Table~\ref{tbl:SummaryResult} and Section~\ref{sec:application}.

The source of these improvements is already visible in a representative benchmark instance. Consider the RSD instance with $n=2^{20}$, $k=32771$, and $h=2467$ over $\mathbb{F}_{2^{128}}$. As reported in Table~\ref{tbl:ComparingResult1}, the previous best hybrid algebraic attack has an estimated complexity of $2^{114.27}$ field operations, whereas our attack requires only $2^{84.50}$. 
Notably, the two attacks guess almost the same number of zero coordinates: the previous attack uses 11,819 zero guesses, as reported in Table~5 of \cite{AC:WWYLYZ25}, while our optimized attack uses 11,826, as reported in Table~\ref{tbl:SummaryResult}.
Thus, the improvement comes primarily from how these guesses are distributed across the regular blocks, rather than from the number of guesses itself. For the corresponding zero-guessing distributions, the probability that all guesses are correct is approximately $2^{-62.18}$ for the previous attack and $2^{-40.37}$ for our balanced strategy; the general success-probability analysis is given in Section~\ref{ssec:algorithmNonbinary}. This difference alone accounts for about $21.8$ of the $29.8$ bits of complexity reduction. The remaining gain comes from the reduced algebraic cost resulting from balanced variable elimination, refined preprocessing, and a tighter treatment of the final overdetermined system.

\subsubsection{Concrete security consequences}
We reassess RSD instances used as concrete hardness assumptions in seven cryptographic applications. In every application, either all or a subset of the parameter sets targeting 128-bit security fall below their claimed level under our estimates, with security losses reaching up to 66 bits. For the refined 128-bit parameter sets of~\cite{EC:LWYY24}, none achieves the claimed security level: the estimated losses range from 27 to 47 bits over non-binary fields and from 17 to 48 bits over binary fields.

The impact of our refinements is also evident on benchmark RSD instances. For parameter sets from~\cite{CCS:BCGI18,EC:LWYY24} that have repeatedly served as benchmarks in subsequent cryptanalytic work~\cite{EC:BriOyg23,EPRINT:KimLee24,C:EssSan24,AC:WWYLYZ25}, our estimates lower the previous best attack complexity by up to 29 bits over non-binary fields and 25 bits over binary fields. The refined attack gives the lowest complexity for all listed benchmark instances except two binary cases. In particular, although R-ISD previously outperformed algebraic attacks on 10 of the 14 binary benchmark parameters, the refined algebraic attack is best on 12 of the 14 parameters considered here.

Table~\ref{tbl:SummaryResult} summarizes these concrete complexity estimates together with the corresponding optimized algorithmic parameters. Section~\ref{sec:application} gives the full comparison and discusses the implications for the affected constructions.


These results shift the emphasis from proposing another RSD solver to revising the current understanding of the concrete hardness of RSD. In the parameter regimes considered here, the gap between previously reported attacks and the refined analysis is large enough to invalidate all or part of the claimed secure parameter sets in each application we examine.

\begin{table}[htbp]
	\centering
    \scriptsize
    \caption{
    Summary for complexity estimates for the various parameter sets and corresponding algorithmic parameters.
    The algorithmic parameters are specified in Section~\ref{sec:algorithmNonbinary}, while the symbols $\wr$, $\dagger$, and $\ddagger$ are defined in Section~\ref{sec:application}.
    }
    \label{tbl:SummaryResult}
        \begin{tabular}{
                >{\centering}p{1.5cm} |
                >{\centering}p{1.1cm} >{\centering}p{1.1cm} >{\centering}p{1.1cm} |
                >{\centering}p{1.3cm} | 
                >{\centering}p{1.1cm} >{\centering}p{3.1cm}
            }
            \hline
            Source 
            & $n$ & $k$ & $h$
            & $q$
            & our & $(\eta_{1},\eta,\eta_{\fix})$\tabularnewline\hline\hline
            \multirow{8}{*}{\makecell{Stationary \\ SD~\cite{C:KPRR25}}}
            & $2^{16}$ & $2^{15}$ & 69  & \multirow{4}{*}{$2^{128}$} & 97.90 & (11011,28235,4477) \tabularnewline
            & $2^{14}$ & $2^{13}$ & 74  & & 95.57 & (3948,7092,1069) \tabularnewline
            & $2^{12}$ & $2^{11}$ & 90  & & 101.90 & (1301,1768,233) \tabularnewline
            & $2^{10}$ & $2^{9}$  & 107 & & 104.13 & (339,404,46)    \tabularnewline\cline{2-7}
            & $2^{16}$ & $2^{15}$ & 69  & \multirow{4}{*}{2} & 97.58 & (13687,28156,4487) \tabularnewline
            & $2^{14}$ & $2^{13}$ & 74  & & 94.52 & (4398,7007,1080) \tabularnewline
            & $2^{12}$ & $2^{11}$ & 90  & & 96.46 & (1277,1667,244) \tabularnewline
            & $2^{10}$ & $2^{9}$  & 107 & & 76.40$^{\dagger}$ & (229,286,57) \tabularnewline\hline
            \multirow{9}{*}{\makecell{PCGs \\ Beaver \\ \cite{AC:MMST25} }}
            & $2^{20}$ & $2^{18}$ &       47 & \multirow{9}{*}{2} &  73.97 & (262086,262086,0) \tabularnewline
            & $2^{21}$ & $2^{18}$ &       31 & &  64.01 & (262107,262107,0) \tabularnewline
            & $2^{21}$ & $2^{19}$ &       45 & &  76.07 & (524221,524221,0) \tabularnewline
            & $2^{22}$ & $2^{19}$ &       30 & &  66.76 & (524250,524250,0) \tabularnewline
            & $2^{22}$ & $2^{20}$ &       44 & &  78.58 & (1048523,1048523,0) \tabularnewline
            & $2^{23}$ & $2^{21}$ &       42 & &  80.67 & (2097073,2097073,0) \tabularnewline
            & $2^{24}$ & $2^{22}$ &       41 & &  83.05 & (4194243,4194243,0) \tabularnewline
            & $2^{25}$ & $2^{23}$ &       39 & &  85.02 & (8388563,8388563,0) \tabularnewline
            & $2^{26}$ & $2^{24}$ &       37 & &  86.99 & (16777173,16777173,0) \tabularnewline\hline
            \multirow{5}{*}{\makecell{ Committed \\ VOLE \\ \cite{CCS:SLYYWW25,eprint:SLYYWW25}}}
            & 1456 & 760 & 91 & \multirow{5}{*}{$2$} 
            & 90.89$^{\dagger}$ & (484,585,83) \tabularnewline
            & 1472 & 762 & 92 &
            & 90.70$^{\dagger}$ & (483,585,84) \tabularnewline
            & 1488 & 765 & 93 &
            & 90.66$^{\dagger}$ & (483,586,85) \tabularnewline
            & 1272 & 433 & 159 & 
            & 61.90$^{\dagger}$ & (134,198,75) \tabularnewline
            & 1296 & 434 & 162 &
            & 61.02$^{\dagger}$ & (128,194,77) \tabularnewline\hline
            \multirow{2}{*}{\makecell{Ferret \\ \cite{CCS:YWLZW20}}}
            & 609728 & 36288 & 1269 & \multirow{2}{*}{$2$} 
            & 98.38 & (5296,17832,16578) \tabularnewline
            & 10805248 & 589760 & 1319 &
            & 110.23 & (16035,299181,289259) \tabularnewline\hline
            \multirow{2}{*}{\makecell{Wolverine \\ \cite{SP:WYKW21}}}
            & 642048 & 19870 & 2508 & \multirow{2}{*}{$2^{61}\!\!-\!\!1$}
            & 82.78 & (230,7221,12648) \tabularnewline
            & 10805248 & 589760 & 1319 &
            & 110.23 & (6958,300538,289221) \tabularnewline\hline 
            \multirow{1}{*}{\makecell{{\tiny TinyKeys}\cite{JC:HOSS22}}}
            & 1280 & 860 & 80 & 2
            & 118.60$^{\ddagger}$ & (664,726,53)\tabularnewline\hline
            \multirow{1}{*}{\makecell{{\tiny MPCitH}~\cite{AC:BCCGJ24}}}
            & 1736 & 776 & 217 & 2
            & 125.63$^{\dagger}$ & (434,484,74)\tabularnewline\hline
        \end{tabular}

        \begin{tabular}{
                >{\centering}p{1.5cm} |
				>{\centering}p{0.7cm} >{\centering}p{0.7cm} >{\centering}p{0.7cm} | 
                >{\centering}p{1.0cm} >{\centering}p{2.5cm} |
                >{\centering}p{1.0cm} >{\centering}p{2.5cm}
			}
            \multicolumn{8}{c}{\vspace{-3pt}}\\ 
            \hline
            \multirow{2}{*}{Source} 
            & \multirow{2}{*}{$n$}
            & \multirow{2}{*}{$k$}
            & \multirow{2}{*}{$h$}
            & \multicolumn{2}{c|}{$q=2^{128}$} 
            & \multicolumn{2}{c}{$q=2$}\\\cline{5-8}
            & & & 
            & our & $(\eta_{1},\eta,\eta_{\fix})$
            & our & $(\eta_{1},\eta,\eta_{\fix})$
            \tabularnewline\hline\hline
            \multirow{14}{*}{\makecell{PCGs \\ \cite{CCS:BCGI18,EC:LWYY24}}}
            & $2^{22}$ & 64770 & 4788 &  86.70 & (0,6904,57849)  &  85.07 & (0,20803,39162) \tabularnewline
            & $2^{20}$ & 32771 & 2467 &  84.50 & (0,11826,20843)  &  79.53 & (0,9307,20895)  \tabularnewline
            & $2^{18}$ & 15336 & 1312 &  93.45 & (2211,7293,6986)&  83.73 & (1220,5944,7023)\tabularnewline
            & $2^{16}$ & 7391  & 667  & 110.45 & (2752,4894,2326)&  99.78 & (2278,4200,2353)\tabularnewline
            & $2^{14}$ & 3482  & 338  & 115.96 & (1746,2592,729)  & 103.64 & (1510,2234,749)  \tabularnewline
            & $2^{12}$ & 1589  & 172  & 120.21 & (984,1251,197)   & 104.25 & (836,1065,211)   \tabularnewline
            & $2^{10}$ & 652   & 106  & 150.32$^{\wr}$ & (509,550,31)     & 113.86$^{\ddagger}$ & (396,434,41)     \tabularnewline\cline{2-8}
            & $2^{22}$ & 64770 & 2735 &  86.51 & (0,0,0)         &  83.10 & (0,35404,25081) \tabularnewline
            & $2^{20}$ & 32771 & 1419 &  80.86 & (0,5403,26014)  &  77.74 & (0,15701,14296) \tabularnewline
            & $2^{18}$ & 15336 & 760  &  74.68 & (0,5359,9272)   &  71.64 & (0,6600,7271)   \tabularnewline
            & $2^{16}$ & 7391  & 389  &  73.97 & (154,4114,3092) &  70.35 & (221,3704,3113) \tabularnewline
            & $2^{14}$ & 3482  & 198  &  77.83 & (790,2357,976)  &  73.56 & (811,2143,992)  \tabularnewline
            & $2^{12}$ & 1589  & 98   &  80.06 & (702,1237,273)   &  74.75 & (680,1127,285)   \tabularnewline
            & $2^{10}$ & 652   & 57   &  89.31 & (450,547,49)     &  78.66$^{\dagger}$ & (400,482,57)     \tabularnewline\hline
		\end{tabular}
        
        \begin{tabular}{
                    >{\centering}p{1.5cm} |
    				>{\centering}p{0.7cm} >{\centering}p{0.7cm} |
                    >{\centering}p{0.7cm} >{\centering}p{1.0cm} >{\centering}p{2.1cm} |
                    >{\centering}p{0.7cm} >{\centering}p{1.0cm} >{\centering}p{2.1cm}
    			}
            \multicolumn{9}{c}{\vspace{-3pt}}\\ 
            \hline
            \multirow{2}{*}{Source} 
            & \multirow{2}{*}{$n$}
            & \multirow{2}{*}{$h$}
            & \multicolumn{3}{c|}{$q=2^{128}$} 
            & \multicolumn{3}{c}{$q=2$}\\\cline{4-9}
            & & 
            & $k$ & our & $(\eta_{1},\eta,\eta_{\fix})$
            & $k$ & our & $(\eta_{1},\eta,\eta_{\fix})$
            \tabularnewline\hline\hline
            \multirow{6}{*}{\makecell{PCGs \\ \cite{EC:LWYY24}}}
            & $2^{22}$ & 4788 
            & 75396 & 87.93 & (0,15018,60361)
            & 80074 & 87.45 & (0,21484,53785)\tabularnewline
            & $2^{20}$ & 2467
            & 30978 & 80.90 & (0,9996,20880)
            & 33354 & 79.88 & (0,9902,20883)\tabularnewline
            & $2^{18}$ & 1312 
            & 14796 & 89.32 & (1643,6738,7001)
            & 15996 & 88.75 & (1910,6622,7005)\tabularnewline
            & $2^{16}$ & 667 
            & 6091 & 89.28 & (1346,3542,2378)
            & 7560 & 102.54 & (2461,4376,2346)\tabularnewline
            & $2^{14}$ & 338 
            & 2909 & 94.93 & (1096,1986,762)
            & 3655 & 109.99 & (1704,2417,739)\tabularnewline
            & $2^{12}$ & 172 
            & 1377 & 100.36 & (726,1022,214)
            & 1657 & 110.53 & (916,1138,206)\tabularnewline\hline
        \end{tabular}
\end{table}

\subsection{Related work}

Briaud and \O ygarden~\cite{EC:BriOyg23} initiated the modern study of algebraic attacks on RSD. Their analysis characterizes the witness degree of the corresponding polynomial systems and shows that algebraic techniques can outperform previously known approaches on several concrete parameters. Their attack already uses blockwise balancing for an initial round of zero guesses, but does not combine this balancing with syndrome-based variable elimination.

Subsequent work further developed hybrid algebraic techniques. AGB2.0~\cite{EPRINT:LWYYY25} employs two rounds of guessing and distributes zero guesses across blocks, while Kim and Lee~\cite{EPRINT:KimLee24} combine two guessing stages with algebraic reduction. These approaches improve the concrete cost in different parameter regimes, but either retain the full syndrome system or separate the two guessing stages in a way that limits the reuse of preprocessing work.

Wang et al.~\cite{AC:WWYLYZ25} provide the strongest previous algebraic estimates for many of the instances studied here by combining zero guessing, syndrome-based variable elimination, preprocessing, and a final hybrid algebraic solve. Our analysis starts from this framework and isolates a remaining source of concrete inefficiency: the relevant operations are not jointly balanced across the regular blocks. Correcting this imbalance, together with two smaller refinements in preprocessing and analysis of the final overdetermined system, yields the complexity reductions reported in this paper.

Combinatorial attacks remain an important point of comparison. Among the ISD-style algorithms relevant to our parameter regime, the regular-decoding techniques of Esser and Santini~\cite{C:EssSan24} provide the strongest competing estimates and are therefore included in our concrete comparisons.

\section{Preliminaries}\label{sec:preliminaries}

\subsection{Notation}\label{ssec:notation}

For prime power $q \in \N$, the finite field with $q$ elements is denoted by $\F_q$. For $a \in \N$, let us define $[a]$ as the set $\{1,\ldots,a\}$. For $i \in \Z$ and $h \in \N$, $[i]_h$ (resp. $i \bmod h$) denote the unique integers in $[h]$ (resp. $\{0,1,\ldots,h\!-\!1\}$) which is congruent to $i$ modulo $h$.

Matrices and vectors are denoted by boldface uppercase and lowercase letters, respectively. All vectors are represented as column vectors. Let $\bfA^\sft$ denote the transpose of $\bfA$. In particular, $\bfa^\sft$ is a row vector. We will often write $\bfx = (\bfx_1,\ldots,\bfx_l)$ to indicate that $\bfx$ is a vertical concatenation of $l$ column vectors $\bfx_i$. We also  write $\bfA = [\bfA_1 \vert \cdots \vert \bfA_l]$
to denote that $\bfA$ is a horizontal concatenation of sub-matrices $\bfA_i \in \F_q^{r \times n_i}$ where $n = \sum_{i \in [l]} n_i$. In particular, $\bfA = [\bfa_1 \vert \cdots \vert \bfa_n] \in \F_q^{r \times n}$, which means that $\bfA$ is a horizontal concatenation of column vectors $\bfa_i \in \F_q^r$. For a vector $\bfx = (x_1,\ldots,x_n)$, a matrix $\bfA = [\bfa_1 \vert \cdots \vert \bfa_n]$, and a set $I = \{i_1 < \cdots < i_t\} \subseteq [n]$, we define $\bfx_I = (x_{i_1},\ldots,x_{i_t})$
and $\bfA_I = [\bfa_{i_1} \vert \cdots \vert \bfa_{i_t}]$. Let $\abs{\bfx}$ denote the Hamming weight of $\bfx$, i.e., the number of nonzero coordinates of $\bfx$.

\subsection{The regular syndrome decoding problem}\label{ssec:rsd}

The main interest of this paper is to solve a variant of the syndrome decoding problem defined as follows.
\begin{definition}[Computational RSD problem]\label{def:rsd}
Given $\bfH \in \F_q^{(n-k) \times n}, \bfs \in \F_q^{n-k}$, 
the regular syndrome decoding (RSD) problem, denoted as $\operatorname{RSD}(n,k,h)$,
asks to find $\bfx = (\bfx_1,\ldots,\bfx_h) \in \F_q^n$ 
such that 
$\bfH\bfx = \bfs$ and $\abs{\bfx_i} = 1$ for $i \in [h]$
where $\bfx_i \in \F_q^{n/h}$.
\end{definition}
The case where $q = 2$ is referred to as binary $\operatorname{RSD}(n,k,h)$ to distinguish it from other cases.

The plain syndrome decoding (SD) problem, denoted $\operatorname{SD}(n,k,h)$, asks to find $\bfx \in \F_q^n$ such that $\bfH \bfx = \bfs$ and $\abs{\bfx}=h$. The SD problem is known as NP-complete~\cite{BMT78}. There is a provable security reduction from the SD problem to the RSD problem~\cite[Theorem~1]{C:FenJouRiv22}: given an SD instance $(\bfH, \bft)$, applying a random permutation in the columns of $\bfH$, one can obtain a RSD instance with probability ${n/h \choose 1}^h / {n \choose h}$. 
Therefore, if the error rate $rho:=h/n$ is sufficiently small,
the difficulty of the RSD problem can be guaranteed by the difficulty of the SD problem. However, if the code rate is also small at the same time, it becomes vulnerable to attacks specific to the RSD problem, and this is the parameter regime that we have in mind.

\subsection{Linear algebra over finite fields}\label{ssec:linearalgebra}

\subsubsection{Gaussian elimination}\label{sssec:gauss}

Gaussian elimination is the primary operation used in the algebraic attacks against the RSD problem, and the complexity of algebraic attacks varies significantly depending on the algorithm used to perform Gaussian elimination.

A $m \times N$ matrix reduction takes $\calO(m \!\cdot\! N \!\cdot\! \min\{m,N\})$ operations using the standard Gaussian elimination to reach reduced row echelon form. For more advanced methods like Strassen's algorithm, the complexity becomes $\calO(m^{c_{\textup{mm}}})$ for square matrices of size $m$, and when applied to non-square matrices, it becomes $\calO(\max\{m,N\}^{c_{\textup{mm}}})$, where $c_{\textup{mm}} \approx 2.8$. We do not consider algorithms with a smaller exponent (e.g. $c_{\textup{mm}} = 2.3713$ in \cite{SODA:ADWXXZ25}) since they perform less well due to their huge hidden constants. This is consistent with the choices made in previous RSD-related
works~\cite{CCS:BCGI18,EPRINT:KimLee24,AC:WWYLYZ25}.

If $m$ is sufficiently larger than the size of the finite field $q$, then the striped matrix elimination~\cite{AAM:AHM07} can be a better choice, and it has a complexity of $\calO(m^2/\log_q m)$ row operations, and $\calO(m^2 N / \log_q m)$ in total,
which is asymptotically close to being optimal. In particular, the striped matrix elimination is useful for performing Gaussian elimination over the binary field $\F_2$.

We shall use the notation $\textup{C}_{\textup{ge}}(n,M)$ to denote the minimum complexity of performing Gaussian elimination over $m\times N$ matrix over $\F_q$, $q \neq 2$, i.e. 
\begin{align*}
\textup{C}_{\textup{ge}}(m,N) 
= \min \big\{
mN\min(m,N),\,
\max(m,N)^{c_{\textup{mm}}}.
\big\}
\end{align*}
For the binary field case, it is represented by 
\begin{align*}
\textup{C}_{\textup{ge}}^{\textup{bin}}(m,N)
= \min \big\{
mN\min(m,N),\,
\max(m,N)^{c_{\textup{mm}}},\,
m^2N/\log_2 m
\big\}. 
\end{align*}

\subsubsection{Invertibility of random matrices}\label{sssec:invertibility}

Let $\operatorname{GL}(m,q)$ be a set of invertible $m \times m$ matrices over $\F_q$.
Then the number of elements in $\operatorname{GL}(m,q)$ is computed as $\abs{\operatorname{GL}(m,q)}
= \textstyle\prod_{i \in [m]} \left( q^m - q^{i-1} \right)$. Thus, the probability that a random $m \times m$ matrix is invertible is 
\begin{align*}
p(m,q) 
= \abs{\operatorname{GL}(m,q)} / q^{m^2}
= \textstyle\prod_{i \in [m]} \left( 1 - q^{i-m-1} \right).
\end{align*}
For a fixed $q$ and increasing $m$, the probability $p(m,q)$ converges to $p(\infty,q) := \prod_{i=1}^{\infty} \left( 1 - q^{-i} \right)$ and it can be expanded as 
\begin{align*}
p(\infty,q)
= \textstyle\sum_{i = -\infty}^{\infty}
(-1)^i q^{-\frac{i(3i-1)}{2}}
\end{align*}
by the Euler's Pentagonal theorem. As $q$ increases, $p(\infty,q)$ will increase monotonically to $1$. For example, $p(\infty,2) \approx 0.2889$, $p(\infty,3) \approx 0.5605$, 
and $p(\infty,13) \approx 0.9170$.

\subsubsection{Overdetermined linear system}\label{ssec:overdetermined}

We examine how much overdetermined a random system of linear equations would be to have no solution. Note that the following is a more refined analysis of the result presented in \cite{SOSA:BDT22} over general finite fields.

Assume that $m \ge N$. The image of a matrix $\bfM \in \F_q^{m \times N}$ contains a \emph{random} vector $\bfd \in \F_q^m$ with probability $q^{\operatorname{rank}(\bfM)- m}$. Since there are $q^{N-\operatorname{rank}(\bfM)}$ solutions when the linear system $\bfM\bfz = \bfd$ is consistent, the expected number of solutions $\alpha$ is bounded as
\begin{align*}
\alpha \le q^{\operatorname{rank}(\bfM)- m} 
\cdot q^{N-\operatorname{rank}(\bfM)} = q^{N-m}.
\end{align*}
On the other hand, the number of $m \times N$ matrices over $\F_q$ of rank $N$ is 
$\prod_{i \in [N]} \left( q^m - q^{i-1} \right)$, and so the probability that a random $m \times N$ matrix is full-rank $\epsilon$ is
\begin{align*}
\epsilon
= \textstyle \prod_{i \in [N]} \left( q^m - q^{i-1} \right) / q^{mN}
= \prod_{i \in [N]} \left( 1 - q^{i-m-1} \right)
\ge (1-q^{N-m})^N
\end{align*}
which converges to $1 - Nq^{-(m-N)} + \calO(N^2 q^{-2(m-N)})$ as $N \rightarrow \infty$. Note that $Nq^{-(m-N)} \ll 1$ in our parameter regime.

To summarize, we will set $m - N \ge \xi(N)$ so that a random matrix $\bfM$ is full-rank with overwhelming probability and the expected number of solutions of the linear system $\bfM \bfz = \bfd$ for a random $\bfb$ is upper-bounded by $q^{-\xi(N)}$. In particular, if the matrix is full-rank, then there is at most one solution. For example, when $q$ is exponentially large such as $q = 2^{128}$, we take $\xi(N) = c_{\textup{excess}}$ for some constant $c_{\textup{excess}}$.

\section{Refined hybrid algebraic attack over non-binary fields}\label{sec:algorithmNonbinary}

We first describe the refinements at a high level and explain why they reduce the concrete cost of the existing hybrid algebraic framework.

\subsection{Overview}\label{ssec:overview}

The refined attack addresses the regular syndrome decoding problem by transforming it into a quadratic multivariate system and subsequently applying a linearization technique. For the parameter regimes of interest, the resulting system is typically underdetermined for linearization: the number of monomials substantially exceeds the number of available equations. The hybrid strategy addresses this gap by guessing selected coordinates to be zero, thereby reducing the number of variables before linearization.

Let $\beta = \frac{n}{h} \in \mathbb{Z}$ and $\omega = \frac{n-k}{h} \in \mathbb{Z}$. Let $(\bfH,\bfs)$ be an instance of the $\textup{RSD}(n,k,h)$ problem. It can be translated into a quadratic multivariate system by exploiting the regularity condition. This initial quadratic system $\calF$ consists of $n$ variables and $n-k + h \binom{\beta}{2}$ equations as
\begin{align*}
    \bfH\bfx=\bfs,
    \qquad
    x_{i_1}x_{i_2}=0
\end{align*}
where $i_{1} \neq i_{2}$ belong to the same block.

The refined attack starts with a first zero-guessing step, in which $\mu_{1}$ variables in each block are guessed to be zero. The syndrome equation $\bfH\bfx = \bfs$ is used to eliminate $\omega$ variables from each block, for a total of $n-k$ variables. Next, $\nu$ variables in each block are selected as irrelevant variables for the second zero-guessing step. We then compute linear combinations of given $h\binom{\beta-\mu_{1}}{2}$ polynomials in which all monomials contain at most one irrelevant variable; it can be done by reducing the original polynomials in reduced row echelon form with an appropriate monomial ordering. Note that this type of preprocessing also has been used in previous algebraic attacks~\cite{NuTMic:JV17,SOSA:BDT22,AC:WWYLYZ25,EPRINT:LWYYY25}. The resulting system has $h\binom{\beta-\mu_{1}}{2} - hv(hv+1)/2$ equations of the form
\begin{align*}\textstyle
\sum_{i \in S_1}\! a_i\, x_i
+ \sum_{i \in S_1, j \in S_2}\! b_{i,j}\, x_i x_j
+ \sum_{j \in S_2}\! c_j\, x_j
+ \sum_{j_1, j_2 \in S_2}\! d_{j_1,j_2}\, x_{j_1} x_{j_2}
= e
\end{align*}
where $S_1$ is the set of positions of irrelevant variables, and $S_2$ is the remainder.

Finally, the attack performs a second zero-guessing step, in which additional $\mu_{2}$ variables in each block are guessed to be zero. This produces a final system consisting of $k-h(\mu_{1}+\mu_{2})$ variables and $h\binom{\beta-\mu_{1}-\mu_{2}}{2} - hv(hv+1)/2$ equations of the form
\begin{align*}\textstyle
\sum_{i \in S_1}\! a_i\, x_i
+ \sum_{i \in S_1, j \in S_2'}\! b_{i,j}\, x_i x_j
+ \sum_{j \in S_2'}\! c_j\, x_j
+ \sum_{j_1, j_2 \in S_2'}\! d_{j_1,j_2}\, x_{j_1} x_{j_2}
= e
\end{align*}
where $S_2'$ is a subset of $S_2$ of size $\abs{S_2}-h\mu_2$. If the final system has at least as many equations as monomials, then it can be solved by linearization. Any candidate solution obtained in this way is then verified against the original RSD instance.

The reduction in concrete complexity comes from two key observations. On the one hand, after the first zero-guessing step, the variables eliminated by the linear equations are distributed as evenly as possible across the blocks, rather than being concentrated in only a few of them. This balanced allocation increases the success probability of the second zero-guessing step and, by the AM--GM inequality, is favorable for minimizing the overall guessing cost. Figure~\ref{fig:blockVisualization} provides a visual comparison of the differences from the previous methods. 
On the other hand, we refine the complexity accounting of the preprocessing stage. In particular, we identify operations within the preprocessing procedure that can be performed at essentially negligible cost and explicitly account for this fact in our complexity estimate. This does not modify the underlying preprocessing algorithm, but yields a tighter estimate of its actual algebraic cost. Together with the balanced blockwise allocation, this refined accounting contributes to the improved concrete complexity obtained by our attack.

\begin{figure}[t]
    \centering
    \includegraphics[width=0.55\textwidth]{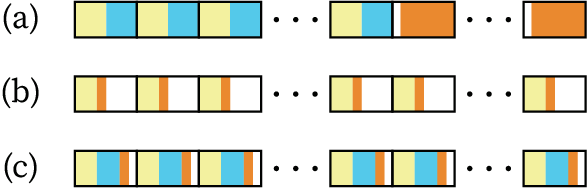}
    \caption{
        Visualization of the blockwise distribution of the first zero guesses (yellow), the variables eliminated by the linear equations (blue), and the second zero guesses (orange) in~\cite{AC:WWYLYZ25},~\cite{EPRINT:LWYYY25}, and our balanced refinement.
        In~(a), the two zero-guessing stages are separated and the linear eliminations are then filled from the front.
        In~(b), the zero guesses are distributed almost evenly without accounting for the linear eliminations.
        In~(c), all components are balanced across the blocks.
    }
    \label{fig:blockVisualization}
\end{figure}

We next extend the balanced construction to the case where the number of components is not a multiple of the number of blocks.
This general case is discussed in the following section.

\subsection{Balanced strategy for guessing and eliminating variables}
\label{ssec:blockNonbinary}

At a high level, we distribute the operations that remove variables as evenly as possible across the regular blocks. In particular, the zero guesses and the variables eliminated through the syndrome equations are assigned so that the number of removed variables differs by at most one between blocks. As a result, the remaining candidate sets stay nearly balanced across blocks, which improves the aggregate guessing probability without sacrificing the number of equations available for the subsequent algebraic solve.

Let $\beta = n / h \in \Z$. We use the following parameters to describe the blockwise allocation of guessed, eliminated, and retained variables. Let $\eta$ be the number of variables that are guessed to be $0$ in the first and second guessing stage. Let $\eta_1$ be the number of variables that are guessed to be $0$ in the first guessing stage. Let $\eta_2 = \eta - \eta_1$ be the number of variables that are guessed to be $0$ in the second guessing stage. Let $\eta_{\fix}$ be the number of variables that are free from being guessed as $0$ in the first and second guessing stage. 

Let $\bfx = (x_1,\ldots,x_n)$. We then call $x_i$ the variable at position $i$. The set $[n]$ is partitioned as $[n] = \dot{\bigcup}_{i \in [h]} \sfB_i$ where $\sfB_i = \{(i-1)\beta+1, \ldots, i\beta\}$, and each block $\sfB_i$ is also partitioned as $\sfB_i = \dot{\bigcup}_{j \in [5]} \sfB_{i,j}$ where
\begin{itemize}
	\item
	$\sfB_{i,1}$: set of positions for the variables that are eliminated in the first guessing stage
	\item
	$\sfB_{i,2}$: set of positions for the variables that are eliminated by the given linear equations
	\item
	$\sfB_{i,3}$: set of positions for the variables that are not the subject of the second guessing stage
	\item
	$\sfB_{i,4}$: set of positions for the variables that are eliminated in the second guessing stage
	\item
	$\sfB_{i,5}$: set of positions for the remaining variables
\end{itemize}
Let us define $\bar{\sfB}_j  = \cup_{i \in [h]} \sfB_{i,j}$ for $j \in [5]$.

\subsubsection{Distributing linear equations.}
We describe a way to distribute variables removed by the linear system to each block.
When the underlying field is non-binary,
we have $n-k$ linear equations such as $\bfH \bfx = \bft$.
Let us define $\omega =  \floor{(n-k)/h}$ and $f_{\textup{lin}} = n-k \bmod h$. Then each of the first $f_{\textup{lin}}$ blocks has $\omega+1$ variables to be eliminated by the linear equations, and each of the remaining $h - f_{\textup{lin}}$ blocks has $\omega$ variables to be eliminated by the linear equations.
Thus, we have
\begin{align*}
\abs{\sfB_{i,2}}
= \left\{\begin{array}{l@{\quad\textup{for}\quad}l}
\omega+1 & i \in I_{\textup{lin}} :=  \big\{1,\ldots, f_{\textup{lin}} \big\} \\
\omega & i \in [h] \setminus I_{\textup{lin}}
\end{array}\right. .
\end{align*} 

\subsubsection{Guessing strategy}
We describe a way to distribute variables guessed as 0 to each block.
Let $\mu := \floor{\eta/h}$, and $f := \eta \bmod h$.

\vspace{1mm}
\noindent
{\it First stage of zero guessing.}
Let us define $\mu_1 = \floor{\eta_1 / h}$, $f_1 = \eta_1 \bmod h$, and $I_{1} = \big\{ [f_{\textup{lin}} \!+\! 1]_h, \ldots,  [f_{\textup{lin}} \!+\! f_1]_h \big\}$.
The block $\sfB_{i}$ for $i \in I_{1}$
has $\mu_1+1$ positions indicating variables guessed as $0$, 
and each of the remaining $h-f_1$ blocks has $\mu_1$ positions indicating variables guessed as $0$, i.e.,
\begin{align*}
\abs{\sfB_{i,1}} = \left\{\begin{array}{l@{\quad\textup{for}\quad}l}
\mu_1 + 1 & i \in I_{1} \\
\mu_1 & i \in [h] \setminus I_{1}
\end{array}\right. .
\end{align*}

\noindent
{\it Second stage of zero guessing.}
Let us define $\mu_2 = \floor{\eta_2 / h}$, and $f_2 = \eta_2 \bmod h$.
Each of the blocks indexed by numbers in 
$I_{2} := \big\{ [f_{\textup{lin}} \!+\! f_1 \!+\! 1]_h, \ldots,  [f_{\textup{lin}} \!+\! f_1 \!+\! f_2]_h \big\}$
has $\mu_2+1$ variables guessed as $0$, 
and each of the remaining $h-f_2$ blocks has $\mu_2$ variables guessed as $0$, i.e.,
\begin{align*}
\abs{\sfB_{i,4}} = \left\{\begin{array}{l@{\quad\textup{for}\quad}l}
\mu_2 + 1 & i \in I_{2} \\
\mu_2 & i \in [h] \setminus I_{2}
\end{array}\right. .
\end{align*}

\noindent
{\it Zero guessing in total.}
We claim that 
\begin{align*}
\abs{\sfB_{i,1}} + \abs{\sfB_{i,4}}
=  \left\{\begin{array}{l@{\quad\textup{for}\quad}l}
\mu + 1 & i \in I := \big\{ [f_{\textup{lin}} \!+\! 1]_h, \ldots, [f_{\textup{lin}} \!+\! f]_h \big\}  \\
\mu & [h] \setminus I
\end{array}\right. .
\end{align*}
If $0 \le f_1 + f_2 \le h$, then 
$I_{1} \cap I_{2} = \emptyset$ and
\begin{align*}
\big( \abs{\sfB_{i,1}}, \abs{\sfB_{i,4}} \big) 
= \left\{\begin{array}{l@{\quad\textup{for}\quad}l}
(\mu_1 + 1,  \mu_2) & i \in I_{1}  \\
(\mu_1, \mu_2 + 1) & i \in I_{2} \\
(\mu_1, \mu_2) & i \in [h] \setminus (I_{1} \cup I_{2})
\end{array}\right. .
\end{align*} 
In case of $f_1 + f_2 = 0$, $\abs{\sfB_{i,1}} = \mu_1$, $\abs{\sfB_{i,4}} = \mu_2$ for all $i \in [h]$, 
$f = 0$, $I = \emptyset$, and 
$\abs{\sfB_{i,1}} + \abs{\sfB_{i,4}} = \mu = \mu_1 + \mu_2
\quad\textup{for}\quad i \in [h]$.
In case of $0 < f_1 + f_2 < h$, the claim holds for $\mu = \mu_1 + \mu_2$, $f = f_1 + f_2$.
In case of $f_1 + f_2 = h$, $f = 0$, $I = \emptyset$ and
$\abs{\sfB_{i,1}} + \abs{\sfB_{i,4}}
=  \mu = \mu_1 + \mu_2 + 1 \quad\textup{for}\quad i \in [h]$.
If $f_1 + f_2 > h$, 
then $I_{1} \cap I_{2} \neq \emptyset$, $I_1 \cup I_2 = [h]$, and
\begin{align*}
\big( \abs{\sfB_{i,1}}, \abs{\sfB_{i,4}} \big) 
= \left\{\begin{array}{l@{\quad\textup{for}\quad}l}
(\mu_1+1,\mu_2+1) & i \in I_{1} \cap I_{2} \\
(\mu_1+1,\mu_2) & i \in I_{1} \setminus I_{2} \\
(\mu_1, \mu_2+1) & i \in I_{2} \setminus I_{1} 	
\end{array}\right. .
\end{align*}
In this case, the claim holds for $\mu = \mu_1 + \mu_2 +1$, $f = f_1 + f_2 - h$, $I = I_1 \cap I_2$.

\subsubsection{Distributing preprocessed variables}
Let us define $\nu = \floor{\eta_{\fix} / h}$ and $f_{\fix} = \eta_{\fix} \bmod h$.
Let $I_{\fix} = \big\{ [f_{\textup{lin}}+f+1]_h, \ldots, [f_{\textup{lin}}+f+f_{\fix}]_h \big\}$.
Each of the blocks whose position belongs to the set $I_{\fix}$
has $\nu+1$ variables free from being guessed as 0,
and each of the remaining $h - f_{\fix}$ blocks has $\nu$ variables free from being guessed as 0, i.e.,
\begin{align*}
\abs{\sfB_{i,3}} 
= \left\{\begin{array}{l@{\quad\textup{for}\quad}l}
\nu+1 & i \in I_{\fix} \\
\nu & i \in [h] \setminus I_{\fix}
\end{array}\right. .
\end{align*}

\subsubsection{Requirement for block size}
Clearly, there is one condition that the block size must satisfy: 
the sum of the cardinality of the first four sub-blocks must not exceed $\beta$,
i.e., $\sum_{j =1}^{4} \abs{\sfB_{i,j}} \le \beta$ for all $i \in [h]$.

Let $a_1 := \floor{(n-k + \eta)/h}$ and $b_1 := (n-k) + \eta \bmod h$.
By using similar arguments used above, we can show that
\begin{align*}
\abs{\sfB_{i,2}} + (\abs{\sfB_{i,1}} \!+\! \abs{\sfB_{i,4}})
= \left\{\begin{array}{l@{\quad\textup{for}\quad}l}
a_1 + 1 & i \in J_1 := \big\{1,\ldots,b_1\} \\
a_1 & i \in [h] \setminus J_1
\end{array}\right.
\end{align*}
where $a_1 \in \{\omega + \mu, \omega + \mu + 1\}$.
Moreover, letting $a_2 := \floor{(n-k + \eta + \eta_{\fix})/h}$,
$b_2 := b_1 + f_{\fix} \bmod h$,
we also can show that 
\begin{align*}
\abs{\sfB_{i,1}} + \abs{\sfB_{i,2}} + \abs{\sfB_{i,3}} + \abs{\sfB_{i,4}}   
= \left\{\begin{array}{l@{\quad\textup{for}\quad}l}
a_2 + 1 & i \in J_2 := \big\{1,\ldots,b_2\big\} \\
a_2 & i \in [h] \setminus J_2
\end{array}\right.
\end{align*}
where $a_2 \in \{ a_1 + \nu, a_1 + \nu + 1 \}$.
By combining these, we can reach at the inequality
$\sum_{j =1}^{4} \abs{\sfB_{i,j}} \le a_1 + \nu + 2 \le \omega + \mu + \nu + 3$.
More precisely, we have 
$\max_{i \in [h]} {\textstyle\sum_{j=1}^{4}} \abs{\sfB_{i,j}} 
= \omega + \mu + \nu + \ceil{\tfrac{f+f_{\textup{lin}}+f_{\fix}}{h}}
= \ceil{\frac{\eta+(n-k)+\eta_{\fix}}{h}}$.

\subsubsection{Related numbers}

We use the following quantities to summarize the block partitions and the resulting equation counts.
\begin{itemize} 
%
%

	\item    
	$\Delta_1 \hspace{.2mm}
	:= \sum_{i \in [h]} {\abs{\sfB_i \setminus (\sfB_{i,1} \cup \sfB_{i,2})} \choose 2}$ \\[1mm]
	\phantom{aa}\hspace{2.3mm}
	$=  \left\{\begin{array}{l@{\quad\quad\hspace{3mm}\textup{for}\quad}l}
	h {\beta-\omega-\mu_1-1 \choose 2} 
	+ (h \!-\! f_{\textup{lin}} \!-\! f_1)(\beta \!-\! \omega \!-\! \mu_1 \!-\! 1)
	& f_{\textup{lin}} + f_1 \le h \\
	h {\beta-\omega-\mu_1-1 \choose 2}
	- (f_{\textup{lin}} \!+\! f_1 \!-\! h)(\beta \!-\! \omega \!-\! \mu_1 \!-\! 2)
	&  f_{\textup{lin}}+f_1 > h
	\end{array}\right.$.
	\vspace{1mm}
	
	\item
	$\Delta_2 \hspace{.2mm}
	:= \sum_{i \in [h]} {\abs{\sfB_i \setminus (\sfB_{i,1} \cup \sfB_{i,2} \cup \sfB_{i,4})} \choose 2}$ \\[1mm]
	\phantom{aa}\hspace{2.3mm}
	$=  \left\{\begin{array}{l@{\qquad\qquad\hspace{.5mm}\textup{for}\quad}l}
	h {\beta-\omega-\mu-1 \choose 2} 
	+ (h \!-\! f_{\textup{lin}} \!-\! f)(\beta \!-\! \omega \!-\! \mu \!-\! 1)
	& f_{\textup{lin}} + f \le h \\
	h {\beta-\omega-\mu-1 \choose 2}
	- (f_{\textup{lin}} \!+\! f \!-\! h)(\beta \!-\! \omega \!-\! \mu \!-\! 2)
	&  f_{\textup{lin}}+f > h
	\end{array}\right.$.
	\vspace{1mm}

    \item
	$\zeta := \sum_{i \in [h]} {\abs{\sfB_{i,3}} \choose 2}
	= f_{\fix} {\nu+1 \choose 2} + (h-f_{\fix}){\nu \choose 2}$.
	\vspace{1mm}	
	
	\item
	$\delta_1 := \sum_{i \in [h]} {\abs{\sfB_i \setminus (\sfB_{i,1} \cup \sfB_{i,2})} \choose 2}
	- \sum_{i \in [h]} {\abs{\sfB_{i,3}} \choose 2}
	= \Delta_1 - \zeta$.
	\vspace{1mm}
    
	\item
	$\delta_2 := \sum_{i \in [h]} {\abs{\sfB_i \setminus (\sfB_{i,1} \cup \sfB_{i,2} \cup \sfB_{i,4})} \choose 2} -\sum_{i \in [h]} {\abs{\sfB_{i,3}} \choose 2}
	= \Delta_2 - \zeta$.
    \vspace{1mm}
    
    \item
    $\tau := \sum_{i \in [h]} \abs{\sfB_{i,2}} \!\cdot\! \abs{\sfB_{i,4}}$.
\end{itemize}

\subsection{Refined attack procedure}\label{ssec:algorithmNonbinary}

We start with the system $\calF$ of $m := n-k + h {\beta \choose 2}$ equations and $n$ variables,
derived from the given RSD problem.
Concretely, the system $\calF$ consists of the following 
$n-k$ linear equations and $h {\beta \choose 2}$ quadratic equations:
\begin{align} \nonumber
\bfH \bfx &= \bfs, \\ \label{eqn:quadratic}
x_{i_1} x_{i_2} &= 0
\quad\textup{for}\,\, i_1 \neq i_2 \in \sfB_i,\,\, i \in [h]. 
\end{align}
where $\bfH \in \F_q^{(n-k) \times n}$ and $\bfx = (x_1,\ldots,x_n)$.

By eliminating $n-k$ variables using the given linear equations,
we can obtain a system $\calF'$ of $h{\beta \choose 2}$ quadratic equations and $k$ variables. 
If $h{\beta \choose 2} < {k \choose 2} + 2k$, then $\calF'$ is underdetermined from the perspective of linearization attack.
In this case, we can make the system $\calF'$ solvable by guessing some variables as 0. 
Otherwise, we can solve the system using the linearization technique.
Even in this case, if the loss resulting from the probability of guessing is acceptable,
the RSD problem can be solved with less complexity (which rarely happens).

\subsubsection{Step~1. First stage of zero guessing}

The following steps are executed as a nested guessing procedure.
At the outer level, we repeatedly perform the first stage of zero guessing.
In each outer trial, we randomly select a new set
$\bar{\sfB}_1 \subseteq [n]$ of $\eta_1$ positions as described in
Section~\ref{ssec:blockNonbinary}, and set
$\bfx_{\bar{\sfB}_1}={\bf 0}$.
For this first-stage guess, we construct a quadratic system $\calF_1$
by eliminating $n-k$ variables using the given linear equations.
More precisely, for each choice of $\bar{\sfB}_1$, we obtain the
linear system
\begin{align}\label{eqn:lin1}
\bfH_{[n] \setminus \bar{\sfB}_1} \bfx_{[n] \setminus \bar{\sfB}_1} = \bfs.
\end{align}
Next, we randomly select a set $\bar{\sfB}_{2} \subseteq [n] \setminus \bar{\sfB}_{1}$ of $n-k$ positions as described in Sec~\ref{ssec:blockNonbinary}.
By performing Gaussian elimination, 
we transform $\bfH_{[n] \setminus \bar{\sfB}_1}$ into $\bfA$ 
such that the sub-matrix $\bfH_{\bar{\sfB}_2}$ of $\bfH_{[n] \setminus \bar{\sfB}_1}$ is reduced to $\bfI_{n-k}$,
i.e., $\bfQ\bfH_{\bar{\sfB}_2} = \bfI_{n-k}$
where $\bfQ$ denote an invertible matrix satisfying $\bfQ \bfH_{[n] \setminus \bar{\sfB}_1} = \bfA$.
Then, \eqref{eqn:lin1} is transformed into 
\begin{align}\label{eqn:lin2}
\bfA \bfx_{[n]\setminus \bar{\sfB}_1} = \bfQ \bfs.
\end{align}
Thus, writing 
\begin{align*}
\bfQ\bfH_{[n] \setminus (\bar{\sfB}_1 \cup \bar{\sfB}_2)} 
= \big[ \bfa_1 \vert \cdots \vert \bfa_{n-k} \big]^\sft
\,\,\,\textup{for some } \bfa_i \in \F_q^{k-\eta_1}
\textup{ and }
\bfQ\bfs = (t_1,\ldots,t_{n-k}), 
\end{align*}
from \eqref{eqn:lin2}, we have 
\begin{align}\label{eqn:linear}
x_i = t_i - \bfa_i^\sft \,\bfx_{[n] \setminus (\bar{\sfB}_1 \cup \bar{\sfB}_2)}
\quad\textup{for}\quad i \in \bar{\sfB}_2.
\end{align}

Substituting these $n-k$ linear equations into \eqref{eqn:quadratic},
we can obtain a system $\calF_1$ of $m_1$ quadratic equations and $n_1 := k - \eta_1$ variables as follows:
\begin{equation}\label{eqn:f1}
\begin{aligned}
x_{i_1} x_{i_2} &= 0
\textup{ for } 
i_1 \neq i_2 \in \sfB_i \!\setminus\! (\sfB_{i,1} \!\cup\! \sfB_{i,2}), \\
x_{i_1}
\!\big(\bfa_{j_2}^\sft \bfx_{[n] \setminus (\bar{\sfB}_1 \cup \bar{\sfB}_2)} - t_{j_2} \!\big) &= 0
\textup{ for } i_1 \in \sfB_i \!\!\setminus\!\! (\sfB_{i,1} \!\cup\! \sfB_{i,2}), 
j_2 \in \sfB_{i,2}, \\
\big(\bfa_{j_1}^\sft \bfx_{[n] \setminus (\bar{\sfB}_1 \cup \bar{\sfB}_2)} - t_{j_1} \!\big)
\big(\bfa_{j_2}^\sft \bfx_{[n] \setminus (\bar{\sfB}_1 \cup \bar{\sfB}_2)} - t_{j_2} \!\big) &= 0
\textup{ for } j_1 \neq j_2 \in \sfB_{i,2}
\end{aligned}
\end{equation}
where $i$ runs through the set $[h]$.
Then, the number of equations is 
\begin{align*}\textstyle
m_1 
=  \sum_{i \in [h]} {\abs{\sfB_i \setminus \sfB_{i,1}} \choose 2}
= f_1{\beta-\mu_1-1 \choose 2} + (h-f_1){\beta-\mu_1 \choose 2}.
\end{align*}
We call $x_{i_1}x_{i_2}$ (resp. $x_{i_1}x_{i_2} = 0$) \emph{free monomial} (resp. \emph{free equation})
since it can be eliminated from the system $\calF_1$ for almost free.
Note that there are 
$\Delta_1:= \textstyle\sum_{i \in [h]} 
{ \abs{\sfB_i \setminus (\sfB_{i,1} \cup \sfB_{i,2})} \choose 2 }$
free monomials.

Let $\bfz_1$ be a vector of length $N_1 := {n_1 \choose 2} + 2n_1$
such that each coordinate of $\bfz$ corresponds to the monomials of the variables of 
$\bfx_{[n] \setminus (\bar{\sfB}_1 \cup \bar{\sfB}_2)}$ with degree $\le 2$.
Then, under some appropriate ordering of monomials,
we can write the system $\calF_1$ as 
$\bfM_1 \bfz_1 = \bfd_1$ where
$\bfM_1 = \begin{bmatrix}
\bfI_{\Delta_1} & \bfO \\
\bfM_1' & \bfM_1''
\end{bmatrix} \in \F_q^{m_1 \times N_1}$, $\bfd_1 \in \F_q^{m_1}$.
After reducing all free monomials,
the system $\calF_1$ can be represented as
\begin{align*}
\hat{\bfM}_1 \bfz_1 = \hat{\bfd}_1,
\quad
\hat{\bfM}_1
= \begin{bmatrix}
\bfI_{\Delta_1} & \bfO \\
\bfO & \hat{\bfM}_1''
\end{bmatrix}.
\end{align*}

The system $\calF_1$ constructed above is associated with the current
first-stage guess $\bar{\sfB}_1$.
For this fixed first-stage guess, we next perform the preprocessing in
Step~2 and then try multiple second-stage guesses in Step~3.
If these trials do not yield a valid solution, we return to Step~1,
choose a new set $\bar{\sfB}_1$, and repeat the procedure.
Consequently, Step~1 forms the outer level of the guessing procedure,
whereas Step~3 forms an inner level whose trials reuse the preprocessing
obtained in Step~2.

\begin{remark}
In Step~1, we assume that $\bfH_{\bar{\sfB}_2}$ is invertible. This assumption does not pose a significant issue, since a random square matrix is invertible with constant probability. More precisely, as discussed in Section~\ref{sssec:invertibility}, the probability that a uniformly random $(n-k)\times(n-k)$ matrix over $\F_q$ is invertible is $ p(n-k,q)=\prod_{i\in[n-k]}\left(1-q^{\,i-n-k-1}\right)$.
This probability is close to $1$ when $q$ is sufficiently large; in particular, it is essentially $1$ for $q=2^{128}$. For $q=2$, it approaches approximately $0.2889$, which is also consistent with the experiments reported in \ref{app:assumptions}. Similar experimental observations are given in \cite[Section~C.1]{EPRINT:WWYLYZ25}. Since the invertibility probability is bounded below by a positive constant, conditioning on this event changes the expected running time by at most a constant factor.
Accordingly, as is standard in ISD-type algorithms~\cite{C:EssSan24,EC:LWYY24} and previous algebraic attacks~\cite{EPRINT:KimLee24,AC:WWYLYZ25}, we omit this constant factor from our complexity estimates.
\end{remark}

\subsubsection{Step~2. Preprocessing via Gaussian elimination}

For each first-stage trial in Step~1, we perform preprocessing once
on the resulting system $\calF_1$ before trying the corresponding
second-stage zero guesses.
This preprocessing is specific to the current first-stage guess
$\bar{\sfB}_1$, but its result is shared by all second-stage guesses
considered under this first-stage trial.
To this end, we randomly select a set
$\bar{\sfB}_3 \subseteq
[n] \setminus (\bar{\sfB}_1 \cup \bar{\sfB}_2)$
of $\eta_{\fix}$ positions as described in
Section~\ref{ssec:blockNonbinary}.
These positions correspond to variables that are excluded from the
variables considered for zero guessing in the second stage.



Let $\bfx_{\fix} := \bfx_{\bar{\sfB}_3} \in \F_q^{\eta_{\fix}}$
and $\bfx_{\textup{var}} := \bfx_{[n] \setminus (\bar{\sfB}_1 \cup \bar{\sfB}_2 \cup \bar{\sfB}_3)} 
\in \F_q^{k - \eta_1 - \eta_{\fix}}$.
Then, after performing Gaussian elimination to $\bfM_1$, 
the system $\calF_1$ can be represented by the following equations:
\begin{align*}
x_{i_1} x_{i_2} &= 0 
\quad\hspace{5.6mm}
\textup{for } i_1 \neq i_2 \in \sfB_i \!\setminus\! (\sfB_{i,1} \!\cup\! \sfB_{i,2}), \\
b_{i_1,i_2}\, x_{i_1} x_{i_2}
+ \bfc_{i_1,i_2}^\sft \bfx_{\fix} 
+ \bfx_{\fix}^\sft \bfD_{i_1,i_2} \bfx_{\textup{var}} & \\[-1.5mm]
+\, \bfx_{\textup{var}}^\sft \bfE_{i_1,i_2} \bfx_{\textup{var}}
+ \bfe_{i_1,i_2}^\sft \bfx_{\textup{var}}&= d_{i_1,i_2}'
\quad\textup{for } i_1 \in \sfB_{i,3}, i_2 \in \sfB_{j,3}, i \neq j, \\[.5mm]
b_{i_1,i_1}\, x_{i_1}^2 
+ \bfc_{i_1,i_1}^\sft \bfx_{\fix} 
+ \bfx_{\fix}^\sft \bfD_{i_1,i_1} \bfx_{\textup{var}} & \\[-1.5mm]
+\, \bfx_{\textup{var}}^\sft \bfE_{i_1,i_1} \bfx_{\textup{var}}
+ \bfe_{i_1,i_1}^\sft \bfx_{\textup{var}}&= d_{i_1}'
\quad\hspace{3.3mm}\textup{for } i_1 \in \sfB_{i,3}, \\[.5mm]
\bfc_{l}^\sft \, \bfx_{\fix}
+ \bfx_{\fix}^\sft \bfD_{l} \bfx_{\textup{var}}
+ \bfx_{\textup{var}}^\sft \bfE_{l} \bfx_{\textup{var}}
+ \bfe_{l}^\sft \bfx_{\textup{var}}&= d_{l}'
\quad\hspace{4.6mm}\textup{for } l \in [m_1']
\end{align*}
where $b_{i_1,i_2}, b_{i_1,i_2} \in \{0,1\}$, $\bfE_{i_1,i_2}, \bfE_{i_1,i_1}, \bfE_{l}$ are upper triangular,
and $i, j$ runs through $[h]$.
Here $m_1'$ denote the number of equations 
that are not free and contain no quadratic monomials of variables of $\bfx_{\fix}$.
Since there are at most $\gamma :=  {\eta_{\fix} \choose 2} + \eta_{\fix}$ equations 
that contain quadratic terms of variables of $\bfx_{\fix}$,
letting $\delta_1$ be the number of equations 
that are free but are not consists only of the variables of $\bfx_{\fix}$,
we have 
\begin{align*}
m_1' \ge m_1 - \gamma -\delta_1
\end{align*}
where $\delta_1 := \big\vert  
\big\{ x_{i_1}x_{i_2} : i_1 \neq i_2 \in \sfB_i \!\setminus\! (\sfB_{i,1} \!\cup\! \sfB_{i,2}) 
\textup{ for } i \in [h] \big\}
\!\setminus\! \big\{ x_{i_1}x_{i_2} : i_1 \neq i_2 \in \sfB_{i,3} 
\textup{ for } i \in [h] \big\} 
\big\vert$ and
the equality holds when $b_{i_1,i_2} = b_{i_1,i_1} = 1$ for all $i_1, i_2$.

Once this preprocessing has been completed, the resulting system is
reused for multiple trials of the second-stage zero guessing in Step~3.
Thus, for a fixed first-stage guess $\bar{\sfB}_1$, Step~2 is performed
once, while Step~3 may be performed multiple times with different choices
of $\bar{\sfB}_4$.
If none of these second-stage trials produces a valid solution, the
algorithm returns to Step~1 and starts a new outer trial with a different
choice of $\bar{\sfB}_1$.  The preprocessing in Step~2 is then recomputed
for the new system $\calF_1$.

\subsubsection{Step~3. Second stage of zero guessing}

For the fixed first-stage guess and the preprocessing obtained in
Steps~1 and~2, respectively, we repeatedly try second-stage zero guesses.
In each inner trial, we select a new admissible set
$\bar{\sfB}_4$ of $\eta_2$ positions and set
$\bfx_{\bar{\sfB}_4}={\bf 0}$.
The preprocessed system obtained in Step~2 is reused across these
second-stage trials.
We focus on solving the subsystem $\calF_1'$ of $\calF_1$
consisting only of the fourth type of equations:
\begin{align}\label{eqn:f1sub}
\bfc_{l}^\sft\, \bfx_{\fix}
+ \bfx_{\fix}^\sft \bfD_{l} \bfx_{\textup{var}}
+ \bfx_{\textup{var}}^\sft \bfE_{l} \bfx_{\textup{var}}
+ \bfe_{l}^\sft\, \bfx_{\textup{var}}&= d_{l}'
\quad\textup{for } l \in [m_1'].
\end{align}
If $m_1'$ is not sufficiently large,
we additionally guess some remaining variables as 0
so that the resulting subsystem is solvable via linearization technique.

In each second-stage trial, we randomly select a new set
$\bar{\sfB}_4$ of $\eta_2$ positions
as described in Section~\ref{ssec:blockNonbinary}, and set
$\bfx_{\bar{\sfB}_4} = {\bf 0}$.
Let $\bfx_{\textup{var}}' = \bfx_{\bar{\sfB}_5}$ where $\bar{\sfB}_5 = [n] \setminus (\cup_{i=1}^{4} \bar{\sfB}_i)$.
Then the new system $\calF_2$ obtained as a result of applying $\bfx_{\bar{\sfB}_4} = {\bf 0}$ to $\calF_1$
consists of the following equations:
\begin{align*}
x_{i_1} x_{i_2} &= 0 
\hspace{6.7mm}
\textup{for } i_1 \neq i_2 \!\in\! \sfB_i \!\setminus\! (\sfB_{i,1} \!\cup\! \sfB_{i,2} \!\cup\! \sfB_{i,4}), \\
b_{i_1,i_2} x_{i_1} x_{i_2}
+ \bfc_{i_1,i_2}^\sft \bfx_{\fix} 
+ \bfx_{\fix}^\sft \bfD_{i_1,i_2}' \bfx_{\textup{var}}' & \\[-1.2mm]
+\, (\bfx_{\textup{var}}')^\sft \bfE_{i_1,i_2}' \bfx_{\textup{var}}'
+ (\bfe_{i_1,i_2}')^\sft\, \bfx_{\textup{var}}' &= d_{i_1,i_2}'
\,\,\textup{for } i_1 \in \sfB_{i,3}, i_2 \in \sfB_{j,3}, i \neq j, \\[.5mm]
b_{i_1,i_1} x_{i_1}^2 
+ \bfc_{i_1,i_1}^\sft \bfx_{\fix} 
+ \bfx_{\fix}^\sft \bfD_{i_1,i_1}' \bfx_{\textup{var}}' & \\[-1.2mm]
+\, (\bfx_{\textup{var}}')^\sft \bfE_{i_1,i_1}' \bfx_{\textup{var}}'
+ (\bfe_{i_1,i_1}')^\sft\, \bfx_{\textup{var}}' &= d_{i_1}'
\hspace{4.6mm}\textup{for } i_1 \in \sfB_{i,3}, \\[.5mm]
\bfc_{l}^\sft\, \bfx_{\fix}
+ \bfx_{\fix}^\sft \bfD_{l}' \bfx_{\textup{var}}' \hspace{6em} & \\[-1.5mm]
+ (\bfx_{\textup{var}}')^\sft
\bfE_{l}' \bfx_{\textup{var}}'
+ (\bfe_{l}')^\sft\, \bfx_{\textup{var}}' &= d_{l}'
\hspace{5.9mm}\textup{for } l \in [m_2']
\end{align*} 
where $b_{i_1,i_2}, b_{i_1,i_1} \in \{0,1\}$,
$\bfE_{i_1,i_2}', \bfE_{i_1,i_1}', \bfE_{l}'$ are upper triangular,
and $i, j$ runs through $[h]$.
Then the number of equations in $\calF_2$ is
\begin{align*}
m_2 = \textstyle\sum_{i \in [h]} {\abs{\sfB_i \setminus (\sfB_{i,1} \cup \sfB_{i,4})} \choose 2}
= f {\beta-\mu-1 \choose 2} + (h-f) {\beta-\mu \choose 2}.
\end{align*}
and the number of free equations is
$\Delta_2 = \sum_{i \in [h]} {\abs{\sfB_i \setminus (\sfB_{i,1} \cup \sfB_{i,2} \cup \sfB_{i,4})} \choose 2}$.
For the same reason as before, we have
\begin{align*}
m_2' \ge m_2 - \gamma - \delta_2
\end{align*}
where $\delta_2 := \big\vert  
\big\{ x_{i_1}x_{i_2} : i_1 \neq i_2 \in \sfB_i \!\setminus\! (\sfB_{i,1} \!\cup\! \sfB_{i,2} \!\cup\! \sfB_{i,4}) 
\textup{ for } i \in [h] \big\}
\!\setminus\! \big\{ x_{i_1}x_{i_2} : i_1 \neq i_2 \in \sfB_{i,3} 
\textup{ for } i \in [h] \big\} 
\big\vert$ and
the equality holds when $b_{i_1,i_2} = b_{i_1,i_1} = 1$ for all $i_1, i_2$.

One thing to note is that $m_2' \le m_1'$.
This is because each equation $\bfc_{l}^\sft\, \bfx_{\fix}
+ \bfx_{\fix}^\sft \bfD_{l}' \bfx_{\textup{var}}'
+ (\bfx_{\textup{var}}')^\sft \bfE_{l}' \bfx_{\textup{var}}'
+ (\bfe_{l}')^\sft\, \bfx_{\textup{var}}' = d_{l}'$
in $\calF_2$ is the result of substituting $\bfx_{\bar{\sfB}_4} = { \bf 0}$ into \eqref{eqn:f1sub}
which a linear combination of equations in \eqref{eqn:f1},
where $x_{i_1} \big( \bfa_{j_2}^\sft \bfx_{[n] \setminus (\bar{\sfB}_1 \cup \bar{\sfB}_2)} - t_{i_2} \big) = 0$
becomes zero equation when $0$ is assigned to $x_{i_1}$ for $i_1 \in \bar{\sfB}_4$.
Thus, after assigning ${\bf 0}$ to $\bfx_{\bar{\sfB}_4}$, 
there are $m_2' = m_1' - \tau$ 
non-zero equations
where $\tau := \textstyle\sum_{i \in [h]} \abs{\sfB_{i,2}} \!\cdot\! \abs{\sfB_{i,4}}$.
It is also verified from the fact that 
$(m_1 - \gamma - \delta_1) - (m_2 - \gamma - \delta_2)
= (m_1 - \Delta_1) - (m_2 - \Delta_2)  
= \sum_{i \in [h]} \abs{\sfB_{i,2}} \cdot \abs{\sfB_{i,4}}$.

\subsubsection{Step~4. Solve a system of equations via linearization}

We now try to solve the subsystem $\calF_2'$ of $\calF_2$ 
consists of the following equations:
\begin{align*}
\bfc_{l}^\sft\, \bfx_{\fix}
+ \bfx_{\fix}^\sft \bfD_{l}' \bfx_{\textup{var}}'
+ (\bfx_{\textup{var}}')^\sft \bfE_{l}' \bfx_{\textup{var}}'
+ (\bfe_{l}')^\sft\, \bfx_{\textup{var}}' = d_{l}'
\quad\textup{for } l \in [m_2'].
\end{align*}
Since the coefficient $\bfc_l$ have already been reduced while performing Gaussian elimination in Step~2,
the system $\calF_2'$ can be represented by the following equations:
\begin{align*}
b_l  x_{l} 
+ \bfx_{\fix}^\sft \bfD_l' \bfx_{\textup{var}}' 
+ (\bfx_{\textup{var}}')^\sft \bfE_l' \bfx_{\textup{var}}'
+ (\bfe_l')^\sft \bfx_{\textup{var}}' 
&= 0
\quad\hspace{1mm}\textup{for } l \in \bar{\sfB}_3 \\
\bfx_{\fix}^\sft \bfD_l' \bfx_{\textup{var}}' 
+ (\bfx_{\textup{var}}')^\sft \bfE_l' \bfx_{\textup{var}}'
+ (\bfe_l')^\sft \bfx_{\textup{var}}' 
&= d_l'
\quad\textup{for } l \in \big[m_2'' \big]
\end{align*}
where $b_l \in \{0,1\}$ and $m_2'' \ge m_2' - \abs{\bar{\sfB}_3} = m_2' - \eta_{\fix}$.

We first solve the subsystem $\calF_2''$ of $\calF_2'$ defined by the second type of equations via the linearization technique.
That is, we write the system $\calF_2''$ as  
\begin{align*}
\bfM_2'' \bfz_2'' = \bfd_2'',
\quad \bfM_2'' \in \F_q^{m_2'' \times N_2''}
\end{align*}
where 
$\bfz_2'' = \big( (x_{i} x_{j})_{i \in \bar{\sfB}_3, j \in \bar{\sfB}_5}, 
(x_{j_1} x_{j_2})_{j_1 \le j_2 \in \bar{\sfB}_5}, (x_j)_{j \in \bar{\sfB}_5} \big)$
is a vector of length $N_2'' := N_2' - \eta_{\fix}$.
To filter out incorrect zero guesses, 
we require that 
\begin{align*}
m_2'' - N_2'' \ge \xi(N_2'')
\end{align*}
so that $\calF_2''$ has no solution with overwhelming probability
when the guessing $\bfx_{\bar{\sfB}_1 \cup \bar{\sfB}_4} = {\bf 0}$ is incorrect.
Once we find a solution $\bfz$ which includes $\bfx_{\textup{var}}'$, 
we compute $\bfx_{\fix}$ by substituting $\bfz$ into the remaining $l$ equations.  
In turn, by substituting $\bfx_{\fix}$, $\bfx_{\textup{var}}'$, 
and $\bfx_{\bar{\sfB}_1 \cup \bar{\sfB}_4} = {\bf 0}$ into \eqref{eqn:linear}, 
we find $\bfx_{\bar{\sfB}_2}$.

The obtained candidate is verified against the original RSD instance.
If it is valid, the algorithm terminates. Otherwise, we return to Step~3
and try another second-stage zero guess while reusing the preprocessing
computed in Step~2. After the prescribed second-stage trials associated
with the current first-stage guess have been exhausted without finding a
valid solution, we return to Step~1 and start a new outer trial.


\begin{remark}
In Step~4, we assume that an incorrect zero guess leads to an inconsistent system $\calF_{2}''$ with overwhelming probability.
The underlying intuition is that an incorrect guess removes the structure associated with the true error support, so that the resulting system behaves essentially like a random overdetermined system. As discussed in Section~\ref{ssec:overdetermined}, such a system is expected to have no solution except with negligible probability.
We further support this assumption with experiments, whose results are reported in \ref{app:assumptions}.
\end{remark}

\subsection{Complexity analysis}\label{ssec:complexityNonbinary}

The transformation
\begin{align*}
\calF
\longrightarrow \calF_1
\longrightarrow \calF_2
\longrightarrow \calF_2'
\longrightarrow \calF_2''
\end{align*}
describes a single path through the attack.
The complete attack repeatedly executes this procedure using a nested
guessing strategy.

At the outer level, Step~1 is repeated with different first-stage
zero guesses $\bar{\sfB}_1$.
For each such first-stage trial, Step~2 constructs preprocessing
specific to the resulting system $\calF_1$.
At the inner level, the preprocessing is then reused while Steps~3
and~4 are repeatedly executed with different second-stage zero guesses
$\bar{\sfB}_4$.
If no valid solution is obtained from the inner trials associated with
the current first-stage guess, the algorithm returns to Step~1 and
starts another outer trial.

This nested structure is important for the complexity analysis.
The first-stage guess is correct with probability
$\textup{p}_{\eta_1}$, so the costs of Steps~1 and~2 are incurred
an expected $\textup{p}_{\eta_1}^{-1}$ times.
On the other hand, a complete first- and second-stage zero guess is
correct with probability $\textup{p}_{\eta}$, so the cost associated
with each second-stage trial is incurred an expected
$\textup{p}_{\eta}^{-1}$ times.
This gives rise to the two different repetition factors appearing in
the total complexity below.

Recall that each system $\calF$ (resp. $\calF_1$, $\calF_2$, $\calF_2'$, $\calF_2''$) 
has $m$ (resp. $m_1$, $m_2$, $m_2'$, $m_2''$) equations, $n$ (resp. $n_1$, $n_2$, $n_2'$, $n_2''$) variables,
and $N$ (resp. $N_1$, $N_2$, $N_2'$, $N_2''$) monomials where 
\begin{equation*}
\begin{array}{c}
m_1 
\!=\! f_1 {\beta - \mu_1 -1 \choose 2} + (h\!-\!f_1){\beta-\mu_1 \choose 2},\quad 
n_1 = k - \eta_1,\quad
N_1 = \tfrac{n_1(n_1+3)}{2}, \\[1mm]
m_2 
\!=\! f{\beta-\mu-1 \choose 2} + (h\!-\!f){\beta - \mu \choose 2},\quad
n_2 = n_2' = n_2'' = k - \eta,\quad
N_2 = \frac{n_2(n_2+3)}{2}, \\[1mm]
m_2' = m_2 - \gamma - \delta_2,\quad
N_2' = N_2 - \gamma - \delta_2, \qquad
m_2'' = m_2' - \eta_{\fix},\quad 
N_2'' = N_2' - \eta_{\fix},
\end{array}
\end{equation*}
and $\gamma = \tfrac{\eta_{\fix}(\eta_{\fix}+1)}{2}$. 
\vspace{2mm}

We can now state the resulting concrete complexity bound.
\begin{theorem}\label{thm:main}
The refined attack with input parameters $(\eta_1, \eta, \eta_{\fix})$ 
satisfying 
\begin{align*}
m_2'' - N_2'' \ge \xi(N_2''),\quad
\eta+\eta_{\fix} \le k,\quad
\ceil{(n\!-\!k+\eta+\eta_{\fix})/h} \le \beta
\end{align*} 
can solve the $\operatorname{RSD}(n,k,h)$ problem in time complexity
\footnote{It is measured in terms of the number of $\F_q$-operations.}
\begin{align*}
\textup{p}_{\eta_1}^{-1} \cdot (\txT_1 + \txT_2) 
+ \textup{p}_{\eta}^{-1} \cdot \txT_3
\end{align*}
where $\textup{p}_{x} := \left(1-\frac{\floor{x/h}+1}{\beta}\right)^{x \bmod h}
\left(1-\frac{\floor{x/h}}{\beta}\right)^{h - (x \bmod h)}$
and
\begin{align*}
\txT_1 
&= \calO\big( 
\textup{C}_{\textup{ge}}(n\!-\!k, n\!-\!\eta_1) 
+ \Delta_1 + (m_1 \!-\! \Delta_1)(n_1 \!+\! 1)^2 
\big), \\
\txT_2 
&= \calO\big( 
\textup{C}_{\textup{ge}}(m_1\!-\!\Delta_1,N_1\!-\!\Delta_1) 
\big), \\
\txT_3 
&= \calO\big( 
m_2'(k\!-\!\eta_1\!+\!1)\eta_2
+ \textup{C}_{\textup{ge}}(m_2'', N_2'') 
\big).
\end{align*}
\end{theorem}

\begin{proof}
We will calculate the complexity of each step in order.

In Step~1, the probability that all $\eta_1$ zero-guessing are correct is 
$\big(1 - \frac{\mu_1+1}{\beta}\big)^{f_1} \big(1 - \frac{\mu_1}{\beta}\big)^{h-f_1}$ 
where $\mu_1 = \floor{\eta_1/h}$ and $f_1 = \eta_1 \bmod h$.
The cost of performing Gaussian elimination 
on $\bfH_{[n] \setminus \bar{\sfB}_1} \in \F_q^{(n-k) \times (n-\eta_1)}$ 
is $\calO( \textup{C}_{\textup{ge}}(n\!-\!k, n\!-\!\eta_1))$.
Since each quadratic equation $x_{i_1}x_{i_2} = 0$ 
for $i_1 \neq i_2 \in \sfB_i \setminus (\sfB_{i,1}\cup\sfB_{i,2})$
can be constructed in a constant time
and the others can be constructed in time $\calO(n_1^2)$,
the time for constructing $m_1$ equations in $\calF_1$ is 
$\calO(\Delta_1 + (m_1 \!-\! \Delta_1)(n_1\!+\!1)^2)$. 
Next, since the computation of the matrix $\hat{\bfM}_1$, i.e., 
removing all free monomials from equations in $\calF_1$,
can be done while constructing equations, 
no additional cost is required.
These explain the term $\textup{p}_{\eta_1}^{-1} \cdot \txT_1$ in total complexity.
\vspace{1mm}

In Step~2, since only Gaussian elimination is performed 
to eliminate quadratic monomials $x_{i_1}x_{i_2}$ for $i_1, i_2 \in \bar{\sfB}_3$
and all free monomials already have been reduced, 
the complexity of this step becomes $\calO(\textup{C}_{\textup{ge}}(m_1\!-\!\Delta_1,N_1\!-\!\Delta_1))$,
which explains the term $\textup{p}_{\eta_1}^{-1} \txT_2$ in total complexity.
\vspace{1mm}

In Step~3, we first compute the probability $\textup{P}_2$ that all $\eta_2$ zero-guessing are correct 
given that all $\eta_1$ zero-guessing are correct. 
Let $E_i$ be the event that all $\abs{\sfB_{i,4}}$ zero-guessings are correct
given that all $\abs{\sfB_{i,1}}$ zero-guessings are correct over $\bfx_{\sfB_i}$.
Let $A_i$ be the event that $\bfx_{\sfB_{i,2} \cup \sfB_{i,3}}$ contains $1$ 
given that all $\abs{\sfB_{i,1}}$ zero-guessings are correct over $\bfx_{\sfB_i}$.
Then
\begin{align*}
\Pr[E_i] 
&= \Pr[A_i] \Pr[E_i \vert A_i] + \Pr[A_i^c] \Pr[E_i \vert A_i^c] \\
&= \tfrac{\abs{\sfB_{i,2}}+\abs{\sfB_{i,3}}}{\beta-\abs{\sfB_{i,1}}} \cdot 1
+ \Big( 1 - \tfrac{\abs{\sfB_{i,2}}+\abs{\sfB_{i,3}}}{\beta-\abs{\sfB_{i,1}}} \Big)
\Big( 1 - \tfrac{\abs{\sfB_{i,4}}}{\beta-\abs{\sfB_{i,1}}-\abs{\sfB_{i,2}}-\abs{\sfB_{i,3}}} \big) \\
&= 1 - \tfrac{\abs{\sfB_{i,4}}}{\beta - \abs{\sfB_{i,1}}}.
\end{align*}
Thus, we have
\begin{align*}
\txP_2 
&= \textstyle\prod_{i \in [h]} \Pr[E_i]
= \prod_{i \in [h]} \Big( 1 - \tfrac{\abs{\sfB_{i,4}}}{\beta - \abs{\sfB_{i,1}}} \Big) \\
&= \textstyle \prod_{i \in [h]} \Big( 1 - \tfrac{\abs{\sfB_{i,1}}+\abs{\sfB_{i,4}}}{\beta} \Big)
/ \prod_{i \in [h]} \Big( 1 - \tfrac{\abs{\sfB_{i,1}}}{\beta} \Big)
= \textup{p}_{\eta} / \textup{p}_{\eta_1}.
\end{align*}

The time for constructing $m_2'$ equations in $\calF_2'$ is
$\calO(m_2' (\eta_{\fix}\eta_2 + (k-\eta_1-\eta_{\fix})\eta_2) + \eta_2)
= \calO(m_2' (k-\eta_1+1) \eta_2 )$
since it consists only of operations to erase monomial terms related to $\eta_2$ guessed variables
and it requires constant-time to remove all zero or linear dependent equations,
which can be done by tracing positions that equations of the form 
$x_{i_1}(\bfa_{j_2}^\sft \bfx_{[n] \setminus \bar{\sfB}_1 \cup \bar{\sfB}_2} - t_{i_2})$ are placed
and removing all equations at these positions.\
Finally, the complexity of solving the linear system $\bfM_2'' \bfz_2'' = \bfd_2''$ 
where $\bfM_2'' \in \F_q^{m_2'' \times N_2''}$ 
is $\calO(\textup{C}_{\textup{ge}}(m_2'',N_2''))$. 
Thus, these explain the term $\textup{p}_\eta^{-1} \txT_3 
= \textup{p}_{\eta_1}^{-1} \cdot (p_{\eta}/p_{\eta_1})^{-1} \cdot \txT_3$
in the total complexity.
\qed
\end{proof}

\begin{corollary}\label{cor: nonamenow}
Take $\xi(N_2'') = c_{\textup{excess}} \in \N$ in Theorem~\ref{thm:main}.
Then, the refined attack with input parameters $(\eta_1, \eta, \eta_{\fix})$ 
satisfying 
\begin{align*}
m_2'' \ge N_2''+ c_{\textup{excess}},\quad
\eta+\eta_{\fix} \le k,\quad
\ceil{(n\!-\!k+\eta+\eta_{\fix})/h} \le \beta
\end{align*} 
can solve the $\operatorname{RSD}(n,k,h)$ problem in time complexity
\begin{align*}
\textup{p}_{\eta_1}^{-1} \cdot (\check{\txT}_1 + \check{\txT}_2) 
+ \textup{p}_{\eta}^{-1} \cdot \check{\txT}_3
\end{align*}
where 
\begin{align*}
\check{\txT}_1 
&= \calO\big( 
\textup{C}_{\textup{ge}}(n\!-\!k, n\!-\!\eta_1) 
+ \Delta_1 
+ (N_{2} \!+\! c_{\textup{excess}} \!+\! \tau \!-\! \Delta_{2})(n_1 \!+\! 1)^2 
\big), \\
\check{\txT}_2 
&= \calO\big( 
\textup{C}_{\textup{ge}}(N_{2} \!+\! c_{\textup{excess}} \!+\! \tau \!-\! \Delta_{2}, N_1\!-\!\Delta_1) 
\big), \\
\check{\txT}_3 
&= \calO\big( 
(N_2' \!+\! c_{\textup{excess}})
(k\!-\!\eta_1\!+\!1)\eta_2
+ \textup{C}_{\textup{ge}}(N_2''+c_{\textup{excess}}, N_2'') 
\big).
\end{align*}
\end{corollary}

\begin{proof}
The refinement is effective when the final system contains enough equations to solve $\calF_2''$, i.e., $m_2'' \ge N_2'' + c_{\textup{excess}}$.
To obtain $\check{m}_2'' := N_2'' + c_{\textup{excess}}$ equations at Step~4,
we use only as many equations as necessary from the previous steps. 
Thus, it is sufficient to use $\check{m}_1, \check{m}_2'$ equations in $\calF_1, \calF_2'$, respectively,
where
\begin{align*}
\check{m}_1 
&= \check{m}_1' + \gamma + \delta_1
=  N_{2} + c_{\textup{excess}} + \tau + \Delta_{1} - \Delta_{2}, \\
\check{m}_1' 
&= \check{m}_2' + \tau
= N_2' + c_{\textup{excess}} + \tau,\quad
\check{m}_2' 
= \check{m}_2'' + \eta_{\fix}
= N_{2}' + c_{\textup{excess}}.
\end{align*}
Thus, with these choices, the claimed complexity can be obtained. 
\qed
\end{proof}


\subsection{Asymptotic consequences}\label{ssec:asymptoticNonbinaryMain}
We briefly state the two main asymptotic consequences of the refined non-binary attack; the complete derivation, including the rounded blockwise guessing exponent, is given in \ref{ssec:asymptoticNonbianry}.
Let $R=k/n$, $\rho=h/n$, and $\beta=1/\rho$. First, without zero guessing, the refined attack is polynomial whenever
\begin{align}\label{eq:nonbinary-poly-main}
R^2\rho n+3R\rho\le 1-\rho.
\end{align}
Indeed, after eliminating the $n-k$ syndrome variables, the system contains $m_2=\frac{n}{2\rho}(1-\rho)$ quadratic equations and $N_2=\frac{R^2n^2+3Rn}{2}$ linearized monomials; the preprocessing removes the same number of equations and monomials, so $m_2\ge N_2$ suffices. Thus, up to lower-order terms, the polynomial-time region is $R^2\rho\le (1-\rho-o(1))/n$, or $R^2\rho\le n^{-1}$ in the usual coarse convention.

Second, if $k\le n-h$, taking $\eta_1=0$, $\eta=k$, and $\eta_{\rm pre}=0$ removes all variables after the balanced zero-guessing step. Each trial is then polynomial, and the total cost is
\begin{align}\label{eq:nonbinary-upper-main}
\mathrm{poly}(n)\cdot 2^{\rho n\Lambda_\beta(R)+O(1)},
\end{align}
where $\Lambda_\beta$ is the exact rounded blockwise exponent defined in \ref{ssec:asymptoticNonbianry}. Under the continuous approximation, this becomes
$\mathrm{poly}(n)\cdot 2^{\rho n\log_2(1/(1-R))}$.

\section{Extension of the refined attack to binary RSD}\label{sec:algorithmBinary}

The binary case follows the same refined framework as the non-binary case, with changes induced by the additional block-sum equations and the binary monomial structure.

\vspace{-2.5mm}
\subsection{Differences from the non-binary case}\label{ssec:difference}

For binary RSD over $\F_2$, regularity implies an additional linear relation in every block: the coordinates in each block sum to one.
Thus, in addition to the syndrome equations, we have $h$ linear equations of the form
$\sum_{j \in \sfB_i} x_j = 1$ for $i \in [h]$.
These additional equations change the number of variables that can be eliminated before the algebraic stage, but otherwise preserve the blockwise balancing strategy.

\vspace{-2.5mm}

\subsubsection{Distributing linear equations}
We describe how the variables eliminated by the linear equations are distributed across the blocks.
In the binary setting, besides the $n-k$ linear equations arising from $\mathbf{H}\mathbf{x}=\mathbf{t}$, there are $h$ additional block-wise linear equations $\sum_{j \in \sfB_{i}} x_{j} = 1$ where $i \in [h]$.
Thus, in contrast to the non-binary setting, the total number of linear equations is $n-k+h$.
Accordingly, the parameters used to distribute these equations are given by
\begin{align*}
    \omega = \left\lfloor \tfrac{n-k+h}{h} \right\rfloor, \qquad 
    f_{\rm lin} = (n-k+h) \bmod h.
\end{align*}
Apart from this change in the values of $\omega$ and $f_{\rm lin}$, the blockwise construction is identical to that in the non-binary case.
In particular, $\sfB_{i,2}$ is defined in the same way, and its cardinality is computed from the corresponding values of $\omega$ and $f_{\rm lin}$ by the same rule.

\subsubsection{Refined attack procedure}
The algorithm for solving the binary RSD problem follows the same overall procedure as in the non-binary case.
Once $\omega$ and $f_{\rm lin}$ are fixed as above, the partitions $\sfB_{i,1},\ \sfB_{i,2},\ \sfB_{i,3},\ \sfB_{i,4},\ \sfB_{i,5}$ are determined exactly as before, and the same argument gives the condition on the block size.

There are, however, two binary-specific differences.
First, the system contains the additional $h$ block-wise linear equations described above.
Second, the number of quadratic monomials is smaller, since over the binary field any quadratic monomial involving a single variable reduces to a linear monomial.
These changes affect the final constraint used to filter out incorrect zero guesses, and hence the function $\xi(N)$ over the binary field differs from its non-binary counterpart.
Further details are given in Corollary~\ref{ssec:complexityNonbinary}.
Since all other steps are unchanged from the non-binary algorithm, we now proceed to the complexity analysis.

	
	
	

\subsection{Complexity analysis}\label{ssec:complexityBinary}
Because the attack structure is unchanged, the complexity analysis mirrors the non-binary case; only the variable counts, monomial counts, and Gaussian-elimination costs differ.
As in the non-binary case, the attack successively transforms the target system $\calF$ into smaller systems:
\begin{align*}
    \calF 
    \longrightarrow \calF_1 
    \longrightarrow \calF_2 
    \longrightarrow \calF_2'
    \longrightarrow \calF_2''
\end{align*}
We recall that each system $\calF$ (resp. $\calF_1$, $\calF_2$, $\calF_2'$, $\calF_2''$) consists of $m$ (resp. $m_1$, $m_2$, $m_2'$, $m_2''$) equations, $n$ (resp. $n_1$, $n_2$, $n_2'$, $n_2''$) variables, and $N$ (resp. $N_1$, $N_2$, $N_2'$, $N_2''$) monomials.
In the binary case, \begin{align*}
    n_1 = k - h - \eta_1,\quad
    n_2 = n_2' = n_2'' = k - h - \eta,\quad
    N_1 = \tfrac{n_1(n_1+1)}{2},\quad
    N_2 = \tfrac{n_2(n_2+1)}{2},
\end{align*} 
and $\gamma = \binom{\eta_{\fix}}{2}$, whereas $m_{1}, m_{2}, m_{2}', m_{2}'', N_{2}', N_{2}''$ are defined exactly as in the non-binary setting.
\vspace{2mm}

The following theorem gives the binary counterpart of Theorem~\ref{thm:main}.
\begin{theorem}\label{thm:main_bin}
The refined attack with input parameters $(\eta_1, \eta, \eta_{\fix})$ 
satisfying 
\begin{align*}
m_2'' - N_2'' \ge \xi(N_2''),\quad
\eta+\eta_{\fix} \le k - h,\quad
\ceil{(n\!-\!k\!+\!h +\eta+\eta_{\fix})/h} \le \beta
\end{align*} 
can solve the binary $\operatorname{RSD}(n,k,h)$ problem in time complexity
\begin{align*}
\textup{p}_{\eta_1}^{-1} \cdot (\txT_1 + \txT_2) 
+ \textup{p}_{\eta}^{-1} \cdot \txT_3
\end{align*}
where $\textup{p}_{x} := \left(1-\frac{\floor{x/h}+1}{\beta}\right)^{x \bmod h}
\left(1-\frac{\floor{x/h}}{\beta}\right)^{h - (x \bmod h)}$
and
\begin{align*}
\txT_1 
&= \calO\big(
(n-k)(n-\eta_{1}) 
+ \textup{C}^{\textup{bin}}_{\textup{ge}}(n\!-\!k, n\!-\!\eta_1\!-\!h) 
+ \Delta_1 + (m_1 \!-\! \Delta_1)(n_1 \!+\! 1)^2 
\big), \\
\txT_2 
&= \calO\big( 
\textup{C}^{\textup{bin}}_{\textup{ge}}(m_1\!-\!\Delta_1,N_1\!-\!\Delta_1) 
\big), \\
\txT_3 
&= \calO\big( 
m_2'(k\!-\!h\!-\!\eta_1\!+\!1)\eta_2
+ \textup{C}^{\textup{bin}}_{\textup{ge}}(m_2'', N_2'') 
\big).
\end{align*}
\end{theorem}
\begin{proof}
    The refined attack proceeds identically to the non-binary case, except for the elimination arising from the additional $h$ linear equations associated with the blocks $\sum_{j \in \sfB_i \setminus \sfB_{i,1}} x_j = 1$ where $i \in [h]$.
    After the first zero-guessing step with $\eta_1$ variables, we obtain $\bfH_{[n] \setminus \bar{\sfB}_1} \in \F_2^{(n-k) \times (n-\eta_1)}$.
    This matrix can be further reduced using the above $h$ linear equations, and the number of operations required for this step is bounded by $(n-k)(\beta-\mu_1)h$.
    Hence, the cost of obtaining $\bfH_{[n] \setminus \bar{\sfB}_1}' \in \F_2^{(n-k) \times (n-\eta_1-h)}$ is $\calO\big((n\!-\!k)(n\!-\!\eta_1)\big)$.
    Consequently, the cost of Gaussian elimination for $\bfH_{[n] \setminus \bar{\sfB}_1}'$ is $\calO\big(\textup{C}^{\textup{bin}}_{\textup{ge}}(n\!-\!k,\,n\!-\!\eta_1\!-\!h)\big)$.
    \qed
\end{proof}


\begin{corollary}\label{cor:binary}
Take $\xi(N_2'') = a N_2''$ in Theorem~\ref{thm:main_bin}.
Then, the refined attack with input parameters $(\eta_1, \eta, \eta_{\fix})$ 
satisfying 
\begin{align*}
    m_2'' \ge  (a + 1) N_2'',\quad
    \eta+\eta_{\fix} \le k - h,\quad
    \ceil{(n\!-\!k\!+\!h+\eta+\eta_{\fix})/h} \le \beta
\end{align*}
can solve the binary $\operatorname{RSD}(n,k,h)$ problem in time complexity
\begin{align*}
    \textup{p}_{\eta_1}^{-1} \cdot (\check{\txT}_1 + \check{\txT}_2) 
    + \textup{p}_{\eta}^{-1} \cdot \check{\txT}_3
\end{align*}
where 
\begin{align*}
    \check{\txT}_1 
    &= \calO\big(
    (n\!-\!k)(n\!-\!\eta_{1})
    \!+\! \textup{C}^{\textup{bin}}_{\textup{ge}}(n\!-\!k, n\!-\!\eta_1\!-\!h)
    \!+\! \Delta_1 
    \!+\! (N_{2} \!+\! a N_{2}'' \!+\! \tau \!-\! \Delta_{2})(n_1 \!+\! 1)^2 
    \big), \\
    \check{\txT}_2 
    &= \calO\big( 
    \textup{C}^{\textup{bin}}_{\textup{ge}}(N_{2} \!+\! a N_{2}'' \!+\! \tau \!-\! \Delta_{2}, 
    N_{1} - \Delta_{1}) 
    \big), \\
    \check{\txT}_3 
    &= \calO\big( 
    (N_{2}' + a N_{2}'')
    (k\!-\!\eta_1\!+\!1)\eta_2
    + \textup{C}^{\textup{bin}}_{\textup{ge}}((a+1)N_{2}'', N_2'') 
    \big).
\end{align*}
\end{corollary}

\begin{proof}
To obtain $\check{m}_2'' := (a+1)N_2''$ equations
we use only as many equations as necessary from the previous steps. 
Thus, it is sufficient to use $\check{m}_1, \check{m}_2'$ equations in $\calF_1, \calF_2'$, respectively,
where
\begin{align*}
\check{m}_1 
&= \check{m}_1' + \gamma + \delta_1
= N_{2} + a N_{2}'' + \tau + \Delta_{1} - \Delta_{2},\\
\check{m}_1' 
&= \check{m}_2' + \tau
= N_{2}' + a N_{2}'' + \tau,\quad
\check{m}_2' 
= \check{m}_2'' + \eta_{\fix}
= N_{2}' + aN_{2}''.
\end{align*}
Thus, with these choices, the claimed complexity can be obtained. 
\qed
\end{proof}

\subsection{Asymptotic consequences}\label{ssec:asymptoticBinaryMain}
The binary case admits the analogous asymptotic picture; full proofs and the rounded guessing analysis are deferred to \ref{ssec:asymptoticBinary}. Let $R=k/n$ and $\rho=h/n$. If $R\le\rho$, the syndrome equations together with the $h$ block-sum equations already remove all asymptotic degrees of freedom. For $R>\rho$, no zero guessing is needed in the polynomial-time region, and the sufficient condition is
\begin{align}\label{eq:binary-poly-main}
(R-\rho)^2\rho n+(R-\rho)\rho\le 1-\rho.
\end{align}
Here the remaining number of variables is $(R-\rho)n$, the number of linearized monomials is $\frac{(R-\rho)^2n^2+(R-\rho)n}{2}$, and the number of quadratic equations is $\frac{n}{2\rho}(1-\rho)$. Consequently, the coarse polynomial-time condition is $(R-\rho)^2\rho\le n^{-1}$.

For $R>\rho$, choosing $\eta_1=0$, $\eta=k-h$, and $\eta_{\rm pre}=0$ removes all variables after the balanced zero-guessing step and gives
\begin{align}\label{eq:binary-upper-main}
\mathrm{poly}(n)\cdot 2^{\rho n\Lambda_\beta(R-\rho)+O(1)}
\end{align}
operations over $\mathbb F_2$. Under the continuous approximation, the dominant exponential term simplifies to
$2^{\rho n\log_2(1/(1-R+\rho))}$.
These statements are sufficient for comparing the asymptotic regimes; the concrete estimates in the applications use the exact finite-parameter formulas rather than these approximations.

For completeness, the continuous bound immediately yields the following useful regimes. If $R$ and $\rho$ are constants, the exponent is $\rho n\log_2(1/(1-R+\rho))$. If $R$ is constant and $\rho=o(1)$, it becomes $\rho n\log_2(1/(1-R))+o(\rho n)$. If $R=o(1)$ and $\rho=o(R)$, then
\begin{align*}
\log_2 \txT_{\rm total}
\le \frac{\rho nR}{\ln 2}+o(\rho nR)+\mathrm{polylog}(n),
\end{align*}
whereas for $\rho=\lambda R$ with a fixed $0<\lambda<1$,
\begin{align*}
\log_2 \txT_{\rm total}
\le \frac{\rho n(1-\lambda)R}{\ln 2}+o(\rho nR)+\mathrm{polylog}(n).
\end{align*}
These expansions are obtained from $-\log_2(1-x)=x/\ln 2+O(x^2)$ and make explicit how the exponent decreases when the regular weight scales below the code rate.

\section{Concrete security reassessment of existing RSD parameters}\label{sec:application}

We now turn to the main concrete consequence of the refined attack: a reassessment of RSD instances used in existing cryptographic constructions~\cite{CCS:BCGI18,C:KPRR25,C:HOSS18,CCS:YWLZW20,SP:WYKW21,AC:BCCGJ24,AC:MMST25,CCS:SLYYWW25}. Rather than viewing the following results merely as a benchmark of a new solver, we use the complexity bounds from Sections~\ref{sec:algorithmNonbinary} and~\ref{sec:algorithmBinary} to determine whether previously selected RSD parameters still meet their intended security targets.

For each RSD instance, an exhaustive search is carried out to determine $(\eta_{1}, \eta, \eta_\fix)$ attaining the minimum complexity. Throughout this search, condition $N_{2}'' q^{-\xi(N_{2}'')} \ll p_{\eta}$ is imposed so that, when zero guessing is incorrect, the resulting system is inconsistent with high probability. For most instances over non-binary fields, this condition is satisfied by taking $\xi(N_2'') = c_{\textup{excess}} = 1$, while for most binary instances it is satisfied by taking $\xi(N_2'') = a N_2''$, $a=\tfrac{1}{2}$. However, some parameter sets require a different choice of $\xi(N)$ in order to satisfy this condition. In Table~\ref{tbl:SummaryResult}, the symbol $\wr$ denotes complexities computed over non-binary fields with $c_{\textup{excess}}=2$, whereas the symbols $\dagger$ and $\ddagger$ denote complexities computed over binary fields with $a=\tfrac{3}{2}$ and $a=\tfrac{5}{2}$, respectively. It should also be noted that some RSD instances do not satisfy $n/h \in \mathbb{Z}$. In such cases, we assume that the last $n - h \floor{n/h}$ coordinates are zero, as done in several previous works~\cite{EC:BriOyg23,EC:LWYY24,C:EssSan24,EPRINT:KimLee24,AC:WWYLYZ25}. The complexity estimates and corresponding parameters are summarized in Table~\ref{tbl:SummaryResult}. All values reported in Table~\ref{tbl:SummaryResult} can be reproduced using the supplementary materials~\footnote{\url{https://github.com/ho-chang/RSD_ReEstimate}}.

\subsubsection{PCG/COT/VOLE}
Pseudorandom Correlation Generators provide a general framework for generating large amounts of correlated randomness from short seeds, and have found numerous applications in cryptography. Several PCG-based constructions are designed based on the hardness of the regular-noise Learning Parity with Noise problem,
which naturally gives rise to RSD.

The first table in Table~\ref{tbl:SummaryResult} shows parameter sets, proposed in \cite{C:KPRR25}, intended to provide 128-bit security for stationary syndrome decoding problem that consists of multiple RSD instances sharing a supporter. 
The analysis shows that none of these instances achieves the claimed security. 
More precisely, four non-binary parameter sets exhibit a security loss of 23--32 bits, while four binary parameter sets exhibit a security loss of 30--51 bits.

\cite{AC:MMST25} presents PCGs for multiparty Beaver triples
based on both the primal and dual variants of LPN problem with regular noise,
which give rise to the RSD problem.
All of these RSD instances are evaluated under the proposed attack, 
and the resulting complexity estimates show that 
none of the proposed parameter sets achieves the claimed 128 bits of security, 
with security losses of 41--63 bits.

Sun et al.~\cite{CCS:SLYYWW25} introduce a committed VOLE protocol 
whose security is based on two dual LPN problems or on both dual and primal LPN problems. 
Three pairs of RSD instances are identified in the conference version~\cite{CCS:SLYYWW25} 
and the full version~\cite{eprint:SLYYWW25}, respectively.
Since it is sufficient to analyze only one of the two RSD parameters constituting each pair, 
a total of five RSD parameters is subject to analysis; 
three from \cite{CCS:SLYYWW25} and two from \cite{eprint:SLYYWW25}.
The analysis shows that none of these instances achieves 
the claimed 128-bit security level against the proposed attack.
More precisely, all five instances exhibit a security loss of 37--66 bits.

Ferret~\cite{CCS:YWLZW20} is a correlated OT extension protocol and its security is based on the LPN problem with regular noise, 
which naturally gives rise to the RSD problem.
The results show that 
none of these parameter sets achieves the claimed 128-bit security against the proposed attack.
It should be noted that one of the two instances had already been shown by previous attacks to fall below the claimed security level; for this instance, the proposed attack further improves the best previously known complexity by 28 bits (see Table~\ref{tbl:ComparingResult2}).
For the other instance, the proposed attack also reduces the complexity below the claimed 128-bit security level, corresponding to a security loss of 17 bits.

Moreover, the second and third tables in Table~\ref{tbl:SummaryResult}, 
introduced in~\cite{CCS:BCGI18} and later revisited in~\cite{EC:LWYY24}, have been widely adopted as benchmark instances in subsequent cryptanalytic studies~\cite{EC:BriOyg23,EPRINT:KimLee24,C:EssSan24,AC:WWYLYZ25}.
The result shows that the proposed attack improves upon the best previously known attacks by up to 29 bits in the non-binary case and up to 25 bits in the binary case (see Table~\ref{tbl:ComparingResult1}).
For the refined 128-bit parameter sets from \cite{EC:LWYY24},
none of these instances achieves the claimed security level against the proposed attack, 
with security losses of 27--47 bits over non-binary fields and 17--48 bits over the binary field (third table in Table~\ref{tbl:SummaryResult}).


\begin{table}[t]
    \centering
    \scriptsize
    \caption{Complexity estimates for the parameter sets from~\cite{CCS:BCGI18,EC:LWYY24}.}
    \label{tbl:ComparingResult1}
        \begin{tabular}{
            >{\centering}p{2.3cm} |
            >{\centering}p{0.7cm} |
            >{\centering}p{1.2cm} | >{\centering}p{1.2cm} | >{\centering}p{1.2cm} | >{\centering}p{1.2cm} | >{\centering}p{1.2cm} | 
            >{\centering}p{1.2cm}
        }
        \hline
        $(n,k,h)$ & $q$ 
        & \makecell{AGB\\\cite{EC:BriOyg23}} 
        & \makecell{ISD\\\cite{EC:LWYY24}} 
        & \makecell{R-ISD\\\cite{C:EssSan24}} 
        & \makecell{RP\\\cite{EPRINT:KimLee24}} 
        & \makecell{HAGB\\\cite{AC:WWYLYZ25}} 
        & Ours \tabularnewline\hline\hline
        $(2^{22}, 64770, 4788)$ & \multirow{7}{*}{$2^{128}$} & 110 & 152 & 173 & 151 & 89.27 & \textbf{86.70} \tabularnewline 
        $(2^{20},32771, 2467)$  & & 131 & 155 & 173 & 148 & 114.27 & \textbf{84.50} \tabularnewline 
        $(2^{18},15336, 1312)$  & & 133 & 153 & 160 & 147 & 119.03 & \textbf{93.45} \tabularnewline 
        $(2^{16},7391 , 667 )$  & & 150 & 151 & 160 & 142 & 131.94 & \textbf{110.45} \tabularnewline 
        $(2^{14},3482 , 338 )$  & & 150 & 150 & 153 & 139 & 133.15 & \textbf{115.96} \tabularnewline 
        $(2^{12},1589 , 172 )$  & & 150 & 155 & 148 & 140 & 132.60 & \textbf{120.21} \tabularnewline 
        $(2^{10},652  , 106 )$  & & 179 & 194 & 173 & 159 & 158.21 & \textbf{150.32} \tabularnewline\hline 
        $(2^{22},64770, 2735)$  & \multirow{7}{*}{$2^{128}$} & 104 & 108 & 126 & 105 & 86.51 & \textbf{86.51} \tabularnewline 
        $(2^{20},32771, 1419)$  & & 102 & 107 & 122 & 103 & 83.91 & \textbf{80.86} \tabularnewline 
        $(2^{18},15336, 760 )$  & & 107 & 105 & 117 & 100 & 90.95 & \textbf{74.68} \tabularnewline 
        $(2^{16},7391 , 389 )$  & & 111 & 103 & 113 & 97 & 96.39 & \textbf{73.97} \tabularnewline 
        $(2^{14},3482 , 198 )$  & & 110 & 101 & 107 & 94 & 95.62 & \textbf{77.83} \tabularnewline 
        $(2^{12},1589 , 98  )$  & & 107 & 100 & 101 & 89 & 93.08 & \textbf{80.06} \tabularnewline 
        $(2^{10},652  , 57  )$  & & 111 & 111 & 106 & 98 & 97.61 & \textbf{89.31} \tabularnewline\hline 
        $(2^{22}, 64770, 4788)$ & \multirow{7}{*}{$2$} & 103 & 135 & 140 & 143 & 87.04 & \textbf{85.07} \tabularnewline 
        $(2^{20},32771, 2467)$  & & 125 & 135 & 138 & 139 & 105.44 & \textbf{79.53} \tabularnewline 
        $(2^{18},15336, 1312)$  & & 122 & 129 & 126 & 137 & 108.27 & \textbf{83.73} \tabularnewline 
        $(2^{16},7391 , 667 )$  & & 139 & 125 & 126 & 130 & 119.67 & \textbf{99.78} \tabularnewline 
        $(2^{14},3482 , 338 )$  & & 138 & 122 & 118 & 125 & 118.53 & \textbf{103.64} \tabularnewline 
        $(2^{12},1589 , 172 )$  & & 135 & 121 & 109 & 120 & 113.69 & \textbf{104.25} \tabularnewline 
        $(2^{10},652  , 106 )$  & & 145 & 159 & \textbf{113} & 117 & 115.58 & 113.86 \tabularnewline\hline 
        $(2^{22},64770, 2735)$  & \multirow{7}{*}{$2$} & 103 & 99 & 103 & 102 & 86.92 & \textbf{83.10} \tabularnewline 
        $(2^{20},32771, 1419)$  & & 98 & 97 & 98 & 100 & 81.09 & \textbf{77.74} \tabularnewline 
        $(2^{18},15336, 760 )$  & & 104 & 92 & 93 & 97 & 86.74 & \textbf{71.64} \tabularnewline 
        $(2^{16},7391 , 389 )$  & & 108 & 87 & 90 & 92 & 90.90 & \textbf{70.35} \tabularnewline 
        $(2^{14},3482 , 198 )$  & & 106 & 83 & 84 & 88 & 88.97 & \textbf{73.56} \tabularnewline 
        $(2^{12},1589 , 98  )$  & & 103 & 80 & 78 & 84 & 82.16 & \textbf{74.75} \tabularnewline 
        $(2^{10},652  , 57  )$  & & 101 & 90 & \textbf{76} & 84 & 78.44 & 78.66 \tabularnewline\hline 
    \end{tabular}
\end{table}

\subsubsection{ZKP}
Wolverine~\cite{SP:WYKW21} is an 
zero-knowledge proof protocol for Boolean and arithmetic circuits.
Its offline phase is based on an sVOLE extension protocol 
whose security relies on the LPN problem with regular noise. 
Our results show that none of parameter sets proposed in Wolverine attains the claimed 128 bits of security.
It is worth noting that one of the two instances had already been shown by previous attacks to fall below the claimed security; for this instance, the proposed attack further improves the best previously known complexity by 28 bits (see Table~\ref{tbl:ComparingResult2}).
For the other instance, the proposed attack also reduces the claimed 128-bit security by 17 bits.

\subsubsection{MPC}
TinyKeys~\cite{C:HOSS18} is an MPC framework for the dishonest-majority setting that improves the efficiency of GMW and BMR style protocols, whose security relies on the RSD problem.
Among the parameter sets proposed in the full version \cite{JC:HOSS22} of TinyKeys, 
one RSD instance
is estimated to fall short of the claimed 128-bit security level by 9 bits under the proposed attack.

\subsubsection{Signature from MPC-in-the-Head}
Bui et al.~\cite{AC:BCCGJ24} 
introduce a post-quantum digital signature scheme based on the MPC-in-the-head paradigm.
The scheme relies on the regular syndrome decoding assumption, and its zero-knowledge protocol is designed to verify the correctness of an RSD instance within the MPC-in-the-head framework.
We evaluate the parameter set proposed in~\cite{AC:BCCGJ24}, which was intended to achieve the 128-bit security level.
According to Table~\ref{tbl:SummaryResult}, this parameter set is estimated to provide approximately 2 bits less security against the proposed attack than claimed.

\begin{table}[t]
    \centering
    \scriptsize
    \caption{Complexity estimates for the parameter sets from Wolverine~\cite{SP:WYKW21} and Ferret~\cite{CCS:YWLZW20} that were proposed to provide 128-bit security.}
    \label{tbl:ComparingResult2}
        \begin{tabular}{
                >{\centering}p{0.7cm} |
                >{\centering}p{3.0cm} |
                >{\centering}p{0.7cm} |
                >{\centering}p{1.1cm} | >{\centering}p{1.1cm} | >{\centering}p{1.1cm} | >{\centering}p{1.1cm} | >{\centering}p{1.1cm} | 
                >{\centering}p{1.2cm}
            }\hline
            {\tiny Source}
            &$(n,k,h)$ & $q$ & \makecell{AGB\\\cite{EC:BriOyg23}} & \makecell{ISD\\\cite{EC:LWYY24}} & \makecell{R-ISD\\\cite{C:EssSan24}} & \makecell{RP\\\cite{EPRINT:KimLee24}} & \makecell{HAGB\\\cite{AC:WWYLYZ25}} & Ours \tabularnewline\hline\hline
            \multirow{2}{*}{\!\cite{CCS:YWLZW20}}
            &$(609728,36288,1269)$ & \multirow{2}{*}{$2$} & 143 & 138 & 140 & 144 & 126.59 & \textbf{98.38}\tabularnewline
            &$(10805248,589760,1319)$ & & 164 & 152 & 157 & 155 & 149.80 & \textbf{110.23} \tabularnewline\hline
            \multirow{2}{*}{\!\cite{SP:WYKW21}}
            &$(642048,19870,2508)$ & \multirow{2}{*}{$2^{61}\!\!-\!\!1$} & 126 & 154 & 172 & 147 & 110.78 & \textbf{82.78}\tabularnewline
            &$(10805248,589760,1319)$ & & 165 & 160 & 177 & 155 & 149.48 & \textbf{110.23} \tabularnewline\hline
        \end{tabular}
\end{table}

\section{Conclusion}\label{sec:conclusion}
We revisited the concrete hardness of RSD and showed that substantial improvements are possible without introducing a fundamentally new solving paradigm. By jointly balancing zero guesses and syndrome-based eliminations across the regular blocks, refining the complexity accounting of the shared preprocessing, and sharpening the analysis of the final overdetermined system, we obtain significantly lower hybrid-algebraic complexities for many parameters used in the literature.

The concrete impact is substantial. On standard benchmark instances, the previous best attack complexity drops by up to 29 bits over non-binary fields and 25 bits over binary fields. Across the seven cryptographic applications analyzed here, at least one parameter set in every application falls below its claimed 128-bit target, with estimated losses of up to 66 bits. These results show that parameter selection for RSD-based constructions should account for the refined hybrid-algebraic complexity rather than relying on earlier estimates.

%% file: tex_files/Appendix.tex

\section{Experimental validations}

\subsection{Experimental evidence for heuristic assumptions
}\label{app:assumptions}

Our complexity analysis relies on two heuristic assumptions.
First, the submatrix $\bfH_{\bar{\sfB}_2}$ selected in Step~1 is assumed to be invertible with approximately the probability of a uniformly random square matrix over the underlying field.
Second, when the combined zero guess is incorrect, the final overdetermined system $\calF_2''$ is assumed to be inconsistent with high probability. 

We provide experimental evidence for the heuristic assumptions used in our complexity analysis and verify the algebraic recovery procedure of the refined attack.
Our experiments focus on RSD over the binary field $\mathbb{F}_2$.

This choice is deliberate. For the two heuristic assumptions considered here, the binary field provides a particularly informative test setting.
First, the probability that a random square matrix over $\mathbb{F}_2$ is invertible converges to approximately $0.2889$, whereas over a large field such as $\mathbb{F}_{2^{128}}$ it is essentially one.
Second, for a random overdetermined linear system with an excess of $\xi(N)$ equations, the expected number of solutions is bounded by $q^{-\xi(N)}$.
Hence, the predicted rejection of incorrect zero guesses becomes substantially stronger as the field size increases.
Over $\mathbb{F}_2$, both phenomena therefore remain nontrivial and can be observed empirically over many random instances.

Accordingly, our experiments are intended to test whether the qualitative behavior assumed in the analysis is already visible in this more informative binary setting, rather than to provide an exhaustive empirical validation over all finite fields.
The corresponding large-field estimates are based on the random-matrix and random-system analyses of Section~\ref{sssec:gauss}.

\paragraph{Invertibility in Step~1.}

For $n \in \{100,200,\ldots,1000\}$, we generate 10,000 random RSD instances for each $n$, and use the fixed ratios $k/n=0.2$ and $h/n=0.04$ as an interior operating point for the experiments.
These values lie away from the extreme parameter regimes appearing in Table~\ref{tbl:SummaryResult} and give a fixed regular block size $\beta=n/h=25$.
Keeping these ratios constant allows us to vary $n$ while preserving the same block structure and relative code/error rates, so that the experiments isolate the effect of increasing instance size.
The purpose of this choice is not to represent all parameter sets in Table~\ref{tbl:SummaryResult}, but to provide a stable and computationally manageable setting in which to examine the qualitative behavior of the heuristic assumptions.
The observed probability that $\bfH_{\bar{\sfB}_2}$ is invertible is approximately $0.289$ throughout the tested range, in agreement with the random-matrix prediction $p(\infty,2)\approx 0.2889$.
This supports the constant-factor treatment of this event used in our complexity estimates.

\paragraph{Rejection of incorrect zero guesses.}
We next test the random-system heuristic underlying the rejection of incorrect zero guesses.
As discussed in Section~\ref{ssec:linearalgebra}, when $m_2''-N_2''\geq \xi(N_2'')$, a random overdetermined system is expected to have at most $2^{-\xi(N_2'')}$ solutions on average.
We therefore measure the probability that $\calF_2''$ is inconsistent when constructed from an incorrect zero guess.

We perform two experiments, using 10,000 random instances for each parameter choice.
First, for $\operatorname{RSD}(100,20,2)$ with $(\eta_1,\eta,\eta_{\mathrm{pre}})=(9,9,6)$, we vary $a\in\{0.5,0.6,\ldots,3.0\}$ in $\xi(N)=aN$.
As shown in Fig~\ref{fig:AssumptionExp}~(a), the inconsistency probability increases with $a$ and approaches one as the excess number of equations increases.
Second, we fix $a=1$ and consider $\operatorname{RSD}(n,0.2n,0.04n)$ for $n\in\{50,75,\ldots,300\}$, using the optimized attack parameters for each instance size.
Fig~\ref{fig:AssumptionExp}~(b) shows that the inconsistency probability also increases with $n$ and is close to one for the larger tested instances.

These experiments are consistent with the random-system heuristic used in our complexity analysis.
They do not constitute a proof for the full parameter range considered in Table~\ref{tbl:SummaryResult}, but provide empirical evidence that incorrect zero guesses are rejected with high probability under the parameter choices used by the attack.

\begin{figure}[htpb]
    \centering
    \includegraphics[width=\textwidth]{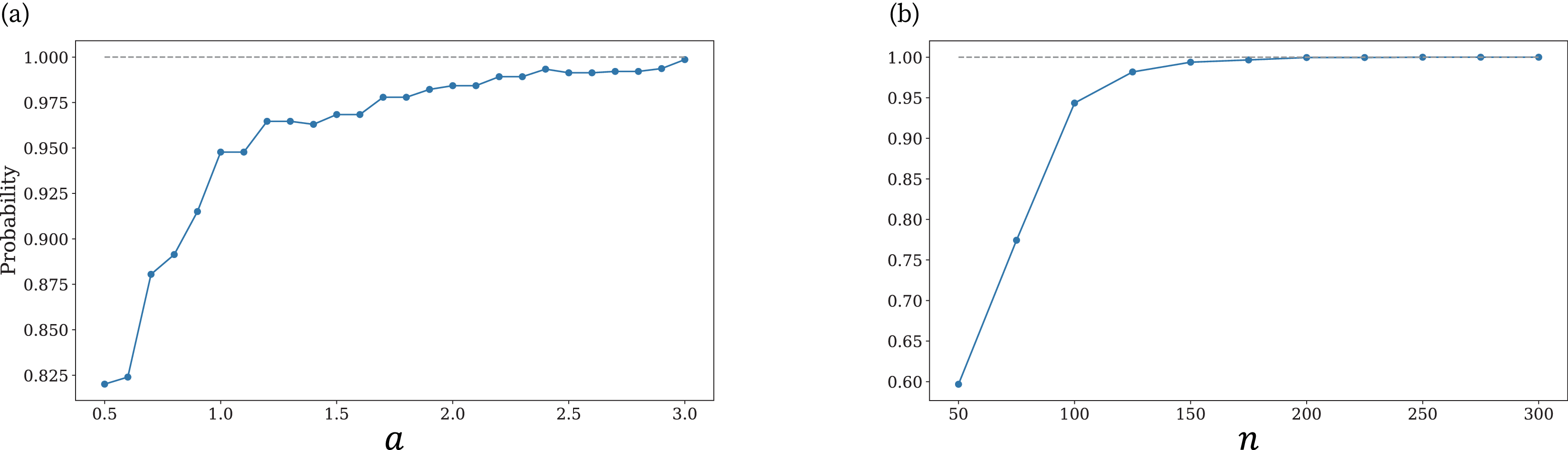}
	\label{fig:AssumptionExp}
	\caption{
        Empirical probability that the final system $\calF_2''$ is inconsistent under an incorrect zero guess, estimated from 10,000 random binary RSD instances per point. (a) $\operatorname{RSD}(100,20,2)$ with $(\eta_1,\eta,\eta_{\mathrm{pre}})=(9,9,6)$, varying $a$ in $\xi(N)=aN$.
        (b) $\operatorname{RSD}(n,0.2n,0.04n)$ for $n\in\{50,75,\ldots,300\}$ with $a=1$, using the optimized attack parameters.
	}
\end{figure}



\subsection{Recovery conditioned on correct zero guesses
}

We next verify that the algebraic stages of the refined attack recover a valid RSD solution when the zero guesses are correct.
To isolate this part of the attack from the potentially expensive search for correct guesses, we condition both guessing stages on a known solution $\mathbf{e}$, choosing $\bar{\sfB}_1$ and $\bar{\sfB}_4$ such that $\mathbf{e}_{\bar{\sfB}_1\cup\bar{\sfB}_4}=0$.
We then execute the remaining algebraic procedure, including syndrome-based elimination, preprocessing, linearization, and reconstruction, and verify the recovered vector against the original RSD instance.

We test 10 independently generated binary $\operatorname{RSD}(512,256,16)$ instances.
For these parameters, the estimated complexity of the full attack is approximately $2^{34.73}$ field operations with $(\eta_1,\eta,\eta_{\mathrm{pre}})=(224,225,14)$.
Under the conditioned correct-guess setting, the attack successfully recovers a valid solution in all 10 instances.
In each case, the recovered vector $\mathbf{x}$ satisfies $\bfH \mathbf{x}=\mathbf{s}$ and has exactly one nonzero coordinate in each regular block.

These experiments provide a direct sanity check of the algebraic recovery procedure under correct zero guesses.
They do not test the cost of finding such guesses, and therefore should not be interpreted as an empirical measurement of the total attack complexity.

\subsection{A toy example of the refined attack}

We illustrate the four stages of the refined attack described in Section~\ref{ssec:algorithmNonbinary}, incorporating the binary modifications from Section~\ref{ssec:difference} and the row budgets from Corollary~\ref{cor:binary}.
All computations are carried out over $\mathbb F_2$, so that $x_j^2=x_j$.
The example follows the sequence
\[
    \mathcal{F}
    \longrightarrow
    \mathcal{F}_1
    \longrightarrow
    \mathcal{F}_2
    \longrightarrow
    \mathcal{F}_2'
    \longrightarrow
    \mathcal{F}_2'',
\]
with the subsystem $\mathcal F_1'$ shared across the second-stage guesses.
For clarity, we display only the successful block partition and second guess.
In the recorded run, three earlier partitions were rejected because the corresponding linear coefficient matrices were singular.

\paragraph{The binary instance.}
Let $(n,k,h)=(18,10,3)$. The coordinates are divided into three blocks
\[
    \sfB_i=\{(i-1)\beta+1, \ldots, i \beta\}, \qquad i\in[h],
\]
which correspond to the three column groups below where $\beta := n/h = 6$.
{\setlength{\arraycolsep}{3pt}\small
\[
 [\,\bfH \mid \bfs \,]=
 \left[\begin{array}{@{}ccc|c@{}}
 \texttt{000101} & \texttt{100011} & \texttt{111001} & \texttt{0}\\
 \texttt{111100} & \texttt{000010} & \texttt{010111} & \texttt{0}\\
 \texttt{111011} & \texttt{011111} & \texttt{111101} & \texttt{1}\\
 \texttt{011111} & \texttt{101101} & \texttt{000001} & \texttt{0}\\
 \texttt{000100} & \texttt{010011} & \texttt{000000} & \texttt{1}\\
 \texttt{000110} & \texttt{000100} & \texttt{000101} & \texttt{1}\\
 \texttt{110110} & \texttt{110100} & \texttt{100101} & \texttt{0}\\
 \texttt{110111} & \texttt{100010} & \texttt{000010} & \texttt{1}
 \end{array}\right],\qquad
 \bfe=\texttt{000010}\,\texttt{000001}\,\texttt{100000}.
\]}%

The known solution satisfies $\operatorname{supp}(\bfe)=\{5,12,13\}$, with exactly one nonzero coordinate in each block.
We use $\bfe$ only to verify the final output; it is not used when choosing the guesses or solving the intermediate systems.

The initial polynomial system is
\begin{equation}\label{app:toy2:initial}
    \bfH \bfx = \bfs,\qquad
    \sum_{j\in \sfB_i}x_j=1,\qquad
    x_u x_v=0
    \quad(u<v,\ u,v\in \sfB_i),\quad i\in[h].
\end{equation}
In the binary setting, the block-sum equations provide additional linear constraints.
Together with the pairwise product equations, they enforce that each block contains exactly one nonzero coordinate.

\paragraph{Parameters and block partition.}
We choose
\[
    (\eta_1,\eta,\eta_{\fix})=(2,4,2),
    \qquad
    \eta_2=\eta-\eta_1=2.
\]
The balanced allocation is determined by
\begin{align*}
    (\omega,f_{\mathrm{lin}})&=(3,2),&
    (\mu_1,f_1)&=(0,2),&
    (\mu,f)&=(1,1),&
    (\nu,f_{\fix})&=(0,2).
\end{align*}

The successful block partition is shown below. The sets
$\bar \sfB_1,\bar \sfB_2$, and $\bar \sfB_3$ are fixed before preprocessing,
whereas $\bar \sfB_4$ and $\bar \sfB_5$ depend on the second-stage guess.

\begin{center}\small
    \setlength{\tabcolsep}{5pt}
    \begin{tabular}{@{}ccllll@{}}
        \toprule
        $i$ & $\sfB_{i,1}$ & $\sfB_{i,2}$ & $\sfB_{i,3}$ & $\sfB_{i,4}$ & $\sfB_{i,5}$\\
        \midrule
        1 & $\{4\}$ & $\{1,2,3,6\}$ & $\{5\}$ & $\varnothing$ & $\varnothing$\\
        2 & $\varnothing$ & $\{7,8,9,12\}$ & $\{10\}$ & $\{11\}$ & $\varnothing$\\
        3 & $\{18\}$ & $\{13,14,16\}$ & $\varnothing$ & $\{17\}$ & $\{15\}$\\
        \midrule
        $|\bar \sfB_j|$ & 2 & 11 & 2 & 2 & 1\\
        \bottomrule
    \end{tabular}
\end{center}

Thus,
\[
    \bar \sfB_1=\{4,18\},\qquad
    \bar \sfB_3=\{5,10\},
\]
while the successful second guess gives
\[
    \bar \sfB_4=\{11,17\},
    \qquad
    \bar \sfB_5=\{15\}.
\]

The corresponding binary monomial counts are
\begin{align*}
    n_1&=k-h-\eta_1=5,
    &N_1&=\textstyle\binom{n_1}{2}+n_1=15,\\
    n_2&=k-h-\eta=3,
    &N_2&=\textstyle\binom{n_2}{2}+n_2=6.
\end{align*}
The numbers of quadratic equations before row selection are
\begin{align*}
    m_1
    &=\sum_{i=1}^{h}\textstyle\binom{\beta-|\sfB_{i,1}|}{2}
      =10+15+10=35,\\
    m_2
    &=\sum_{i=1}^{h}
      \textstyle\binom{\beta-|\sfB_{i,1}|-|\sfB_{i,4}|}{2}
      =10+10+6=26.
\end{align*}

After the first guess, there are $\Delta_1=\binom12+\binom22+\binom22=2$ free monomials, whereas $\Delta_2=0$ after the second guess. Since no two positions in $\bar \sfB_3$ belong to the same block, we have $\delta_1=2$ and $\delta_2=0$.
Hence,
\begin{align*}
    \gamma&=\textstyle\binom{\eta_{\fix}}{2}=1,
    &\tau&=\sum_{i=1}^{h}|\sfB_{i,2}|\,|\sfB_{i,4}|=4+3=7,\\
    N_2'&=N_2-\gamma-\delta_2=5,
    &N_2''&=N_2'-\eta_{\fix}=3.
\end{align*}

The corresponding untruncated row counts used in the complexity analysis are
\[
    m_2'=m_2-\gamma-\delta_2=25,
    \qquad
    m_2''=m_2'-\eta_{\fix}=23.
\]
These quantities should not be confused with the numbers of rows actually
retained in this toy example.

\paragraph{Choosing the row budgets.}
For this small instance, we use a single excess equation, $\xi(N_2'')=1$.
This corresponds to the integer choice $\xi(N)=\lfloor \tfrac{1}{2}N\rfloor$ at $N=3$.

The backward row budgets are therefore
\begin{align*}
    \check m_2''&=N_2''+\xi(N_2'')=4,\\
    \check m_2'&=\check m_2''+\eta_{\fix}=6,\\
    \check m_1'&=\check m_2'+\tau=13,\\
    \check m_1&=\check m_1'+\gamma+\delta_1=16.
\end{align*}
Thus, only $\check m_1-\Delta_1=14$ non-free equations are selected for preprocessing.
The parameter constraints are satisfied:
\[
    \eta+\eta_{\fix}=6\le k-h=7,
\]
and
\[
    \left\lceil
    \frac{n-k+h+\eta+\eta_{\fix}}{h}
    \right\rceil
    =6\le\beta.
\]

\paragraph{Step 1. First stage of zero guessing.}
We first set $\mathbf x_{\bar \sfB_1}=0$.
Combining the syndrome equations with the three block-sum equations yields eleven linear equations in the sixteen remaining coordinates; the three additional equations arise from the binary block-sum constraints.

Let $L$ be the resulting $11\times16$ coefficient matrix and let
\[
    \mathbf b=(\bfs^{\mathsf T},1,1,1)^{\mathsf T}.
\]
For the successful partition, the selected submatrix $\bfL_{\bar \sfB_2}$ is invertible.
Writing $\bfQ = \bfL_{\bar \sfB_2}^{-1}$, we obtain
\begin{align*}
    \mathbf x_{\bar \sfB_2}
    =\bfQ\mathbf b+
    \bfQ \bfL_{[n]\setminus(\bar \sfB_1\cup\bar \sfB_2)}
    \mathbf x_{[n]\setminus(\bar \sfB_1\cup\bar \sfB_2)}.
\end{align*}
This is precisely the binary affine substitution used in Step~1 of the general attack.


Substituting these affine expressions into the pairwise product equations of~\eqref{app:toy2:initial} gives the system $\mathcal F_1$ in the five remaining variables
\[
    (x_5,x_{10},x_{11},x_{15},x_{17}).
\]
The two free equations are
\[
    x_{10}x_{11}=0,
    \qquad
    x_{15}x_{17}=0.
\]
Using the linearization notation of the main text, write $\mathcal F_1$ as
\[
    \bfM_1\mathbf z_1=\mathbf d_1,
\]
where $\mathbf z_1$ contains $N_1=15$ linear and square-free quadratic monomials.

We order the two free monomials first and eliminate their columns, obtaining
\[
    \hat{\bfM}_1=
    \begin{pmatrix}
        \bfI_{\Delta_1}&\bfO\\
        \bfO&\hat{\bfM}_1''
    \end{pmatrix},
    \qquad
    \Delta_1=2.
\]
The lower subsystem has $m_1-\Delta_1=33$ equations and $N_1-\Delta_1=13$ monomial columns.
We select 14$(=\check{m}_{1} - \Delta_{1})$ equations according to the row budget above.
As an example, suppose that Step~1 yields the following system:
{\small\setlength{\arraycolsep}{3pt}
\[
\begin{array}{@{}lll@{}}
x_{17} + x_{10} (x_{15} + x_{17}) = 0, &
x_{11} x_{17} = 0, &
x_{10}(x_{15} + x_{17}) = 0,\\
x_{5}(x_{11} + x_{15} + x_{17}) = 0&
x_{11}(x_{15} + x_{17}) = 0, &
x_{10} (1 + x_{17}) = 0, \\
x_{5}(1 \!+\! x_{15}) \!+\! x_{10}(1 \!+\! x_{15}) \!+\! x_{11} \!+\! x_{15} \!=\! 1, &
x_{5}x_{11} = 0, &
x_{5}x_{10} = 0,\\
x_{10}(1+x_{15} + x_{17}) + x_{11}x_{17} = 0, &
x_{17}(1+x_{5} + x_{10}) = 0, &
x_{17}(x_{10} + x_{11}) = 0\\
x_{15}(1 + x_{5} + x_{10}) = 0, &
x_{10}(x_{5} + x_{15} + x_{17}) = 0
\end{array}
\]}%

\paragraph{Step 2. Preprocessing via Gaussian elimination.}
Fixing $\bar \sfB_3$, define
\[
    \mathbf x_{\mathrm{pre}}
    =(x_5,x_{10})^{\mathsf T},
    \qquad
    \mathbf x_{\mathrm{var}}
    =(x_{11},x_{15},x_{17})^{\mathsf T}.
\]
The only quadratic monomial involving only preprocessing variables is $x_5x_{10}$.
Gaussian elimination isolates this monomial in an equation that also contains the monomial $x_{5}x_{10}$.
The remaining 13$(=\check{m}_{1}')$ equations form the selected shared subsystem $\mathcal F_1'$, with 12$(=N_{1} - \gamma - \delta_{1})$ monomial columns.
Each equation has the form
\begin{equation}\label{app:toy2:shared}
    \mathbf c_l^{\mathsf T}\mathbf x_{\mathrm{pre}}
    +\mathbf x_{\mathrm{pre}}^{\mathsf T}\bfD_l\mathbf x_{\mathrm{var}}
    +\mathbf x_{\mathrm{var}}^{\mathsf T}\bfE_l\mathbf x_{\mathrm{var}}
    +\mathbf e_l^{\mathsf T}\mathbf x_{\mathrm{var}}
    =d_l'.
\end{equation}
After Step~2, we obtain the following system:
{\small\setlength{\arraycolsep}{3pt}
\[
\begin{array}{@{}lll@{}}
x_{17} + x_{10} (x_{15} + x_{17}) = 0, &
x_{11} x_{17} = 0, &
x_{10}(x_{15} + x_{17}) = 0,\\
x_{5}(x_{11} + x_{15} + x_{17}) = 0&
x_{11}(x_{15} + x_{17}) = 0, &
x_{10} (1 + x_{17}) = 0, \\
x_{5}(1 \!+\! x_{15}) \!+\! x_{10}(1 \!+\! x_{15}) \!+\! x_{11} \!+\! x_{15} \!=\! 1, &
x_{5}x_{11} = 0, &
x_{17}(x_{10} + x_{11}) = 0\\
x_{10} x_{15} + x_{11}x_{17} = 0, &
x_{17}(1+x_{5} + x_{10}) = 0, \\
x_{15}(1 + x_{5} + x_{10}) = 0, &
x_{10}(x_{15} + x_{17}) = 0
\end{array}
\]}
Note that, in the above system, there are at most $|\mathbf{x}_{\rm pre}|$ equations satisfying $\bfc_l \neq 0$.
This is achieved during preprocessing.


\paragraph{Step 3. Second stage of zero guessing.}
For the successful second guess, choose
\[
    \bar \sfB_4=\{11,17\}
\]
and set $x_{11}=x_{17}=0$.
The remaining variable vector is
\[
    \mathbf x_{\mathrm{var}}'
    =\mathbf x_{\bar \sfB_5}
    =(x_{15}).
\]
Substitution into $\mathcal F_1$ gives $\mathcal F_2$, while substitution into the selected subsystem $\mathcal F_1'$ gives $\mathcal F_2'$.


In the chosen row basis, 5 of the 13 selected equations become identically zero, leaving eight nonzero equations involving $N_2'=5$ monomials.
Note that the structural quantity $\tau=7$ appearing in the backward row budget is not the number of visibly vanishing rows in this particular row-reduced representation.
Moreover, the eight displayed rows need not be linearly independent; repeated equations are retained.
After Step~3, we obtain the following system:
\begin{align*}
	x_5+x_{15}+x_5x_{15}&=1, &x_{10}&=0,\\
	x_{10}x_{15}&=0,
	&x_5x_{15}&=0,\\
	x_{15}+x_5x_{15}+x_{10}x_{15}&=0
\end{align*}
Note that, in this example, the monomial $x_{10}x_{15} = 0$ remains in four other equations.

\paragraph{Step 4. Solving the final system by linearization.}
Since the previous step yields $\check{m}_{2}''$ equations containing no linear terms in $\mathbf{x}_{\rm pre}$, the remaining subsystem $\mathcal{F}_{2}''$ consists of 4 equations and has $N_{2}''=3$ monomial ccolumns.
Using the monomial ordering of the main text, let
\[
    \mathbf z_2''
    =(x_5x_{15},\,x_{10}x_{15},\,x_{15})^{\mathsf T}.
\]
There is no additional quadratic monomial supported entirely on $\bar \sfB_5$, since $x_{15}^2=x_{15}$ over $\mathbb F_2$.
Retaining $\check m_2''=4$ equations gives
\begin{equation}\label{app:toy2:final}
    \bfM_2''\mathbf z_2''=\mathbf d_2'',
    \qquad
    \bfM_2''=
    \begin{pmatrix}
        1&0&0\\
        0&1&0\\
        0&1&0\\
        1&1&1
    \end{pmatrix},
    \qquad
    \mathbf d_2''=
    \begin{pmatrix}
        0\\0\\0\\0
    \end{pmatrix}.
\end{equation}
The coefficient matrix has rank $3=N_2''$, and the corresponding augmented matrix reduces to
\[
    \left[
    \begin{array}{ccc|c}
        1&0&0&0\\
        0&1&0&0\\
        0&0&1&0\\
        0&0&0&0
    \end{array}
    \right].
\]
Hence,
\[
    \mathbf z_2''=(0,0,0)^{\mathsf T}
\]
is the unique linearized solution.
Substituting this solution into the two recovery equations yields $x_5=1$ and $x_{10}=0$.
The recovered variables satisfy the required mixed-monomial identities, as well as the quadratic equation set aside during Step~2.
Finally, applying the affine substitutions obtained in Step~1 gives
\[
    \mathbf x
    =\texttt{000010}\,\texttt{000001}\,\texttt{100000}
    = \bfe.
\]
Direct verification confirms that
\[
    \bfH\mathbf x=\bfs
    \qquad\text{and}\qquad
    |\mathbf x_{\sfB_i}|=1
    \quad\text{for every }i\in[h].
\]
The final system also has the prescribed excess $\check m_2''-N_2''=1$.

This example illustrates the complete recovery procedure for a successful guess.
It is intended only to demonstrate the mechanics of the refined attack and does not, by itself, establish an overwhelming rejection probability for incorrect guesses.
In particular, consistency of a selected linearized subsystem must not be regarded as a substitute for verifying the solution against the original RSD instance.

\paragraph{System sizes.}
For reference, the following table summarizes the dimensions of the main linearized systems encountered above.
The column counts refer only to coefficient columns and should be distinguished from the untruncated equation counts $m_1,m_2,m_2'$, and $m_2''$ used in the complexity analysis.

\begin{center}\small
\setlength{\tabcolsep}{4pt}
\begin{tabular}{@{}clcc@{}}
\toprule
Step & System or operation & \multicolumn{2}{c}{Rows $\times$ columns}\\
\midrule
1 & Linear system after the first guess & $11\times16$ & $(n-k+h) \times (n-\eta_{1})$\\
1 & Non-free part of $\mathcal F_1$ & $33\times13$ & $(m_{1} - \Delta_{1}) \times (N_{1} - \Delta_{1} + 1)$\\
1 & Selected shared of $\mathcal F_1$ & $14\times13$ & $(\check{m}_{1} - \Delta_{1}) \times (N_{1} - \Delta_{1} + 1)$\\
2 & After preprocessing with selected $\calF_{1}$ & $13\times12$ & $(\check{m}_{1} - \gamma - \delta_{1}) \times (N_{1} - \gamma - \delta_{1} + 1)$\\
3 & $\mathcal F_2'$ & $8\times5$ & $(\check{m}_{2}' + \epsilon) \times N_{2}'$\\
4 & $\calF_2''$ & $4\times3$ & $\check{m}_{2}'' \times N_{2}'' $\\
\bottomrule
\end{tabular}
\end{center}
Note that, in the size of $\mathcal{F}_{2}'$, $\epsilon$ denotes the number of equations that remain after the second zero-guessing step due to the row-budget constraint.
In Step~3, at most $\tau$ equations can remain nonvanishing; therefore, we have $\epsilon \le \tau$.

\section{Asymptotic complexity} 

\subsection{Non-binary fields}\label{ssec:asymptoticNonbianry}

We now analyze the asymptotic complexity of the refined attack over non-binary fields. 
Our analysis follows the same asymptotic viewpoint as
in \cite{{AC:WWYLYZ25}}, but we keep the exact binomial term in the equation count.

Let $R=\frac{k}{n}, \rho = \frac{h}{n}$ and $\beta =\frac{n}{h}=\frac{1}{\rho}$.
We count the number of field operations over $\mathbb{F}_q$. 
Here, we focus on the large-field setting of Corollary~\ref{cor: nonamenow}, where
$\xi(N_2'')=c_{\rm excess}$ for some constant $c_{\rm excess}\in\mathbb{N}$.

For the asymptotic upper bound, it is sufficient to consider the simple choice $\eta_1=0, \eta_{\rm pre}=0$ and $\eta=rn$ for some $0\le r\le R$. 
This means that all zero guesses are made in the second guessing stage. 
This choice is not necessarily optimal for concrete parameters, but it gives a clean asymptotic upper bound. 
With this choice, the
number of remaining variables is $ n_2=k-\eta=(R-r)n$, and hence the number of linearized monomials is 
$N_2=\frac{n_2(n_2+3)}{2}
= \frac{(R-r)^2n^2+3(R-r)n}{2}$.
On the other hand, by the definition of $m_2$ in Section~\ref{ssec:complexityNonbinary}, after $\eta=rn$
balanced zero guesses we have, ignoring only the integer rounding among blocks,
\begin{align*}\textstyle
m_2 = h\binom{\beta-r/\rho}{2}
= \frac{n}{2\rho}(1-r)(1-r-\rho).
\end{align*}
The exact rounded expression is
\begin{align*}\textstyle
m_2 = f\binom{\beta-\mu-1}{2} +(h-f)\binom{\beta-\mu}{2},
\qquad
\mu = \left\lfloor \frac{\eta}{h}\right\rfloor,
\quad
f = \eta\bmod h.
\end{align*}
The continuous expression above is used only to make the asymptotic dependence
on $R$ and $\rho$ explicit.

A key point is that the preprocessing terms cancel in the final
overdeterminedness condition. Indeed, by the definitions in Section~\ref{ssec:complexityNonbinary},
\begin{align*}
m_2' = m_2-\gamma-\delta_2,\,\,
N_2' = N_2-\gamma-\delta_2
\quad\textup{and}\quad
m_2'' = m_2'-\eta_{\fix},\,\,
N_2'' = N_2'-\eta_{\fix}.
\end{align*}
Therefore, it holds $m_2''-N_2'' = m_2-N_2.$
Thus the condition $m_2''-N_2''\ge c_{\rm excess}$ is implied by
\begin{align*}\textstyle
\frac{n}{2\rho}(1-r)(1-r-\rho)
- \frac{(R-r)^2n^2+3(R-r)n}{2}
\ge c_{\rm excess}.
\end{align*}
Equivalently,
\begin{align}\label{eq:nonbinary-asymp-condition}\textstyle
(R-r)^2\rho n
+ 3(R-r)\rho
+ \frac{2c_{\rm excess}\rho}{n}
\le (1-r)(1-r-\rho).
\end{align}

We also record the asymptotic form of the guessing probability. 
Let $c=\left\lfloor r/\rho\right\rfloor$ and $\theta=r/\rho-c$. 
Then the exact rounded exponent of the balanced zero guessing is
$p_\eta^{-1}
= 2^{\rho n\Lambda_\beta(r)+O(1)}$
where
\begin{align}\label{eq:rounded-exponent}\textstyle
\Lambda_\beta(r)
= \theta\log_2\frac{1}{1-(c+1)/\beta}
+(1-\theta)\log_2\frac{1}{1-c/\beta}.    
\end{align}
If the integer rounding is ignored, or if $r\beta$ is an integer, then this
simplifies to
\begin{align}\label{eq:continuous-exponent}
p_\eta = (1-r)^h,
\quad
p_\eta^{-1} = 2^{\rho n\log_2(1/(1-r))}.
\end{align}


\begin{theorem}
\label{thm:nonbinary-asymp-poly}
Let $R = k/n$ and $\rho = h/n$. The refined attack solves
$\mathrm{RSD}(n,k,h)$ over non-binary fields in polynomial time if
\begin{equation}\label{eq:nonbinary-poly-cond}
R^2\rho n + 3R\rho \le 1-\rho
\end{equation}
equivalent to $\frac{n}{2\rho}(1-\rho)\ge \frac{R^2n^2+3Rn}{2}.$
In particular, this condition can be written as $R^2\rho \le \frac{1-\rho-o(1)}{n}.$
Following the coarse asymptotic convention, this is written as $R^2\rho \le n^{-1}.
$
\end{theorem}

\begin{proof}
Take $\eta_1=\eta=\eta_{\fix}=0.$
Then no zero guessing is performed, and hence $p_{\eta_1}=p_\eta=1.$
With this choice, after eliminating $n-k$ variables by the linear equations,
the resulting quadratic system has $m_2 = h\binom{\beta}{2}
= \frac{n}{2\rho}(1-\rho)$
equations, and the number of remaining variables is $n_2 = k = Rn.$
Therefore, the number of linearized monomials is
$N_2 = \frac{n_2(n_2+3)}{2}
= \frac{R^2n^2+3Rn}{2}.$

Among the $m_2$ equations, exactly $\Delta_2$ are free equations, and the corresponding $\Delta_2$ monomials are free monomials. Removing them yields a linearized system with
$m_2-\Delta_2$ equations and $N_2-\Delta_2$ monomials. Hence it is sufficient that
$m_2-\Delta_2 \ge N_2-\Delta_2$,
which is equivalent to $m_2 \ge N_2.$
Substituting the above expressions for $m_2$ and $N_2$, we obtain
\begin{align*}\textstyle
\frac{n}{2\rho}(1-\rho)\ge \frac{R^2n^2+3Rn}{2},
\end{align*}
equivalently, $R^2\rho n + 3R\rho \le 1-\rho$.
%
Under this condition, the resulting linearized system can be solved in
polynomial time. Since all remaining steps consist of Gaussian eliminations and
substitutions on matrices of polynomial size, the total complexity is polynomial
in $n$. \qed
\end{proof}

\begin{theorem}
\label{thm:nonbinary-asymp-upper}
Assume that $k \le n-h$. Then, for the choice $\eta_1=0$, $\eta=k$, and
$\eta_{\fix}=0$, the refined attack solves $\mathrm{RSD}(n,k,h)$ over
non-binary fields in time
$$
\mathrm{poly}(n)\cdot 2^{\rho n\Lambda_\beta(R)+O(1)}
$$
field operations.
\end{theorem}

\begin{proof}
Take $\eta_1=0$, $\eta=k$, and $\eta_{\fix}=0$, equivalently $r=R$. A regular error vector has exactly one nonzero entry in each block of size $\beta$. Hence each block contains exactly $\beta-1$ zero entries, and the whole vector contains exactly $h(\beta-1)=n-h$ zero entries. 
Therefore, the choice $\eta=k$ can be correct only if $k \le n-h.$
Under balanced zero guessing, the $k$ guessed-zero positions are distributed so that each block receives either $\lfloor k/h\rfloor$ or $\lceil k/h\rceil$ guesses. 
Since $k \le h(\beta-1)$,
no block is guessed entirely to zero. In particular, the success probability $p_k$ is positive.

With this choice, $n_2 = k-\eta = 0$ and therefore $N_2 = \frac{n_2(n_2+3)}{2}=0.$
Hence the final linearized system is empty. 
For each fixed zero-guess pattern,
all remaining coordinates are recovered from the linear equations obtained in the reduction step, 
and the resulting candidate is checked against the original
RSD instance. 
Each trial therefore costs only polynomial time.

By \eqref{eq:rounded-exponent}, specialized to $r=R$, the inverse success
probability of the balanced zero guessing is
$$
p_k^{-1}=2^{\rho n\Lambda_\beta(R)+O(1)}.
$$
Multiplying this by the polynomial per-trial cost gives the desired result.
\qed
\end{proof}

Under the continuous approximation in which the blockwise rounding is ignored,
the condition $k \le n-h$ becomes $R \le 1-\rho,$
and the guessing probability satisfies $p_k=(1-R)^h$.
Hence the above bound simplifies to
\begin{align*}
\mathrm{poly}(n)\cdot 2^{\rho n\log_2(1/(1-R))}.    
\end{align*}

\subsection{Binary fields}\label{ssec:asymptoticBinary}

\begin{theorem}\label{thm:binary-asymp-polynomial}
Let $R=k/n$ and $\rho=h/n$. The refined attack solves binary
$\mathrm{RSD}(n,k,h)$ in polynomial time if
$R\le \rho$ or
\begin{equation}\label{eq:binary-poly-exact}
(R-\rho)^2\rho n+(R-\rho)\rho \le 1-\rho.
\end{equation}
\end{theorem}

\begin{proof}
If $R\le \rho$, then $k-h\le 0$. Hence the syndrome equation together with the $h$ block-sum equations already leaves no asymptotic degree of freedom, and the instance can be handled in polynomial time.

Assume now that $R>\rho$, and take
$\eta_1=\eta=\eta_{\rm pre}=0.$
Then no zero guessing is performed, and hence
$p_{\eta_1}=p_\eta=1.$
After eliminating the $n-k+h$ linear equations, the number of remaining variables is
$n_2=k-h=(R-\rho)n.$
Therefore, the number of linearized monomials is
$N_2=\frac{n_2(n_2+1)}{2}
=\frac{(R-\rho)^2n^2+(R-\rho)n}{2}$.
On the other hand, the number of quadratic equations is
$m_2=h\binom{\beta}{2} =\frac{n}{2\rho}(1-\rho).$

Among these, exactly $\Delta_2$ equations are free equations, and the corresponding $\Delta_2$ monomials are free monomials. Removing them yields a linearized system with
$m_2-\Delta_2$ equations and $N_2-\Delta_2$
monomials. Hence it is sufficient that
$m_2-\Delta_2 \ge N_2-\Delta_2,$
which is equivalent to $m_2 \ge N_2.$
Substituting the above expressions gives
\begin{align*}\textstyle
\frac{n}{2\rho}(1-\rho)
\ge
\frac{(R-\rho)^2n^2+(R-\rho)n}{2},
\end{align*}
equivalently,
$(R-\rho)^2\rho n+(R-\rho)\rho \le 1-\rho.$

Under this condition, the resulting linearized system can be solved in
polynomial time. Since all remaining steps consist of Gaussian eliminations and
substitutions on matrices of polynomial size, the total complexity is polynomial
in $n$. \qed
\end{proof}

Under the continuous approximation, \eqref{eq:binary-poly-exact} becomes
$
(R-\rho)^2\rho \le \frac{1-\rho-o(1)}{n},
$
and, in the coarser convention, this is written as
$
(R-\rho)^2\rho \le n^{-1}.
$

\begin{theorem}\label{thm:binary-asymp-upper}
Let $R=k/n$ and $\rho=h/n$, and assume $R>\rho$. Then, for the choice
$\eta_1=0$, $\eta=k-h$, and $\eta_{\rm fix}=0$, the refined attack solves
binary $\mathrm{RSD}(n,k,h)$ in time
\begin{align*}
\mathrm{poly}(n)\cdot 2^{\rho n\Lambda_\beta(R-\rho)+O(1)}
\end{align*}
operations over $\mathbb{F}_2$.
\end{theorem}

\begin{proof}
Take
$\eta_1=0,\eta=k-h,$ and $ \eta_{\rm pre}=0.$
Under balanced zero guessing, the $k-h$ guessed-zero positions are distributed so that each block receives either $\lfloor (k-h)/h\rfloor$ or $\lceil (k-h)/h\rceil$ guesses. Since
$k-h \le n-h = h(\beta-1),$
no block is guessed entirely to zero. In particular, the success probability $p_{k-h}$ is positive.

With this choice, $n_2=k-h-\eta=0,$
and therefore $N_2=\frac{n_2(n_2+1)}{2}=0.$
Hence the final linearized system is empty. For each fixed zero-guess pattern, all remaining coordinates are recovered from the linear equations obtained in the reduction step, and the resulting candidate is checked against the original
RSD instance. Each trial therefore costs only polynomial time.

By \eqref{eq:rounded-exponent}, specialized to $r=R-\rho$, the inverse success probability of the balanced zero guessing is
$p_{k-h}^{-1}=2^{\rho n\Lambda_\beta(R-\rho)+O(1)}$.
Multiplying this by the polynomial per-trial cost gives
$
\mathrm{poly}(n)\cdot 2^{\rho n\Lambda_\beta(R-\rho)+O(1)}.
$ \qed
\end{proof}

Under the continuous approximation, the guessing probability satisfies
$p_{k-h}=(1-R+\rho)^h.$
Hence the above bound simplifies to
\begin{align*}
\mathrm{poly}(n)\cdot 2^{\rho n\log_2(1/(1-R+\rho))}.
\end{align*}

\begin{corollary}
Let $\txT_{\rm total}$ denote the complexity of the refined attack.
Under the continuous approximation used in the proof of
Theorem~\ref{thm:binary-asymp-upper}, the following asymptotic upper bounds
hold. 

\begin{enumerate}
    \item If $R$ and $\rho$ are constants, then
    $
    \txT_{\rm total}
    \le
    \mathrm{poly}(n)\cdot
    2^{\rho n\log_2(1/(1-R+\rho))}.
    $

    \item If $R$ is a constant and $\rho=o(1)$, then
    $
    \txT_{\rm total}
    \le
    \mathrm{poly}(n)\cdot
    2^{\rho n\log_2(1/(1-R))+o(\rho n)}.
    $
    In particular, in coarse $O(\cdot)$ notation,
    $
    \txT_{\rm total}=O\left(2^{\rho n\log_2(1/(1-R))}\right).
    $

    \item If $R=o(1)$ and $\rho=o(R)$, then
    $
    \log_2 \txT_{\rm total}
    \le
    \frac{\rho nR}{\ln 2}
    +o(\rho nR)
    +\mathrm{polylog}(n),
    $
    and hence
    $
    \txT_{\rm total}
    =
    O\left(2^{\rho nR/\ln 2}\right).
    $

    \item If $R=o(1)$ and $\rho=\lambda R$ for a constant $0<\lambda<1$, then
    $
    \log_2 T_{\rm total}
    \le
    \frac{\rho n(1-\lambda)R}{\ln 2}
    +o(\rho nR)
    +\mathrm{polylog}(n),
    $
    and hence
    $
    T_{\rm total}
    =
    O\left(2^{\rho n(1-\lambda)R/\ln 2}\right).
    $
\end{enumerate}
\end{corollary}

\begin{proof}
By above, the leading exponential term under the
continuous approximation is
$
2^{\rho n\log_2(1/(1-R+\rho))}.
$
The first case is immediate.

If $R$ is constant and $\rho=o(1)$, then
$\log_2\frac{1}{1-R+\rho}
= \log_2\frac{1}{1-R}
+O(\rho),
$
which gives the second case.

If $R=o(1)$ and $\rho=o(R)$, then
$\log_2\frac{1}{1-R+\rho} =
\frac{R-\rho}{\ln 2} +O(R^2)
=
\frac{R}{\ln 2}
+o(R),
$
which gives the third case.

Finally, if $\rho=\lambda R$ with $0<\lambda<1$, then
$
\log_2\frac{1}{1-R+\rho}
=
\log_2\frac{1}{1-(1-\lambda)R}
=
\frac{(1-\lambda)R}{\ln 2}
+O(R^2),
$
which proves the fourth case.
\qed
\end{proof}